\documentclass[12pt,a4paper]{article}

\usepackage[T1]{fontenc}
\usepackage[utf8]{inputenc}
\usepackage{float}
\usepackage{graphicx}
\usepackage{subcaption} 
\usepackage[thinlines]{easytable}
\usepackage{hyperref}
\usepackage{wrapfig}
\usepackage{dingbat}
\usepackage{color}
\usepackage{refcount}
\usepackage[table]{xcolor}

\usepackage[toc,page]{appendix}
\usepackage{bbm}
\usepackage{enumitem}
\usepackage{amsmath,amssymb,amsthm}
\usepackage{mathtools}
\usepackage[english]{babel}
\usepackage{natbib} 
\usepackage{xcolor}
\usepackage{booktabs}
\usepackage{multirow}

\theoremstyle{definition}
\newtheorem{theo}{Theorem}[section]
\newtheorem{lem}[theo]{Lemma}

\newtheorem{rem}[theo]{Remark}
\newtheorem{ex}[theo]{Example}
\numberwithin{equation}{section}

\usepackage[plain,noend]{algorithm2e}

\newcommand{\R}{\mathbb{R}}

\newcommand{\E}{\mathbb{E}}
\newcommand{\N}{\mathbb{N}}
\newcommand{\Z}{\mathbb{Z}}

\newcommand{\C}{\mathbb{C}}
\newcommand{\PR}{\mathbb{P}}

\usepackage{bibunits}
\usepackage{algpseudocode}

\defaultbibliographystyle{apalike}
\defaultbibliography{refs}

\begin{document}

\begin{bibunit}

\title{Simultaneous Inference for Partially\\
Observed Functional Time Series}

\author{Patrick Bastian and Tim Kutta}

\maketitle

\begin{abstract}
Functional data analysis (FDA) provides statistical methods for analyzing samples of time-continuous stochastic processes. Measurements often arise in the form of sensor data for a key scientific variable. The practical problem of irregular sensor disruptions has fostered interest in analyzing partially observed random functions.  
Specifically, this paper is motivated by a time series of intermittently missing pollution data with dependence along pollution paths and missingness patterns. To allow statistical analysis, we develop the first inference methods for dependent, partially observed functional time series. Existing methods were not appropriate for this task, because they heavily rely on the independence of the data functions.  
Mathematically, we model data on the space of bounded functions equipped with the supremum norm. This allows simultaneous inference across the entire functional domain, including simultaneous confidence bands -- something existing Hilbert-space-based methods cannot provide. To study non-stationary trends along the time series, we extend state-of-the-art multiscale inference methods (originally developed for scalar data) to partially observed functions. The key application of the latter methods is testing for excessive pollution levels in inner cities.  
Our approach combines state-of-the-art Gaussian approximations with stochastic process theory. Interestingly, it also improves existing results for fully observed functional time series by avoiding a functional CLT.

\end{abstract}

\section{Introduction}

In this article, we develop new statistical methodology for Functional Data Analysis (FDA). Functional data typically arise when a scientific variable is densely observed over a temporal or spatial domain, and its realizations can be interpreted as a coherent curve. Inference for an entire sample of such curves, say $X_1(\cdot),\cdots,X_N(\cdot)$, is the subject of FDA, and we refer to the classical monographs of \cite{bosq:2000} and \cite{HKbook} for an exposition. Examples of functional data include air quality data ($X_i(\cdot)$ is a daily profile of pollutant concentration), medicine ($X_i(\cdot)$ refers to data from a patient's health tracking device) or movement data ($X_i(\cdot)$ refers to a GPS-tracked path). A common problem in the aforementioned examples, and for many more in FDA, is the missingness of values across parts of the domain.  This missingness is often caused by downtime in sensors, due to technical defects or human errors. If missingness is rare and short-lived, imputation may be adequate. However, if there is regular missingness across the data, this needs to be accounted for explicitly to ensure valid inference. 

Various works have developed methodology for partially observed functional data.  An important early work by \cite{kraus:2015} focuses on estimation of principal components and functional parameters when data are missing.
Subsequently, \cite{Kraus:2019} develops tests for the equality of many mean functions across different populations, similar to a functional ANOVA. Missingness in function-on-function regression is studied by \cite{Park:Han:Simpson:2023}, who develop tests for linear hypotheses. Robust analyses are extended to the missing data regime by \cite{Park:Chen:Simpson:2022}. 

Goodness-of-fit tests for the model assumption of missingness at random are developed by \cite{Ofner:Hoermann:Kraus:Liebl:2025} and methods that are robust against systematic missingness are studied by \cite{Liebl:Rameseder:2019}. 
Recently,  \cite{Hudecova:Kirch:2025} have developed stationarity tests for the mean function against different change point alternatives. In their model, changes may occur gradually or abruptly and tests are based on permutation statistics. Specific, more restrictive  missingness schemes have been explored under the label of "functional snippets" (\cite{lin:wang:zhong:2021,lin:wang:2022}). This term refers to functional data, observed on a random interval inside the domain. A challenge in snippet models is covariance estimation, because certain lags in the functional domain are very rarely observed.  

The above literature models the functional observations and the missingness schemes as independent. However, in most time series of sensor data, it is obvious that there exists dependence in both respects.
For example, missingness in one curve $X_i(\cdot)$ may extend into the next curve  $X_{i+1}(\cdot)$, because the sensor outage on day $i$ extends into day $(i+1)$. Conversely, if $X_i(\cdot)$ and $X_{i+1}(\cdot)$ happen to be fully observed, sensor measurements during the $i$th night are continuous with those in the $(i+1)$th morning, again implying dependence.
Such rather natural cases are not accounted for in the literature.  Therefore, our first main contribution is that, to the best of our knowledge, \textit{we develop the first inference methods for partially observed functional data with serial dependence. Dependence may exist both in the functional observations and in the missingness scheme.} As an example, we later study sensor data for air pollutants, where dependence along functions and missingness scheme is self-evident.  Mathematically, dealing with dependence is non-trivial, because many of our key objects are non-measurable in the presence of missingness.

A second aspect of this work is that we treat missingness as a pointwise phenomenon; at each point $t$ in the functional domain data either are or are not available. This may sound banal, yet existing work typically embeds partially observed data into the Hilbert space of square integrable functions $L^2$. Here, pointwise evaluations are  not well-defined, which makes this embedding from our perspective somewhat unnatural.
As an alternative, in this work, we model partially observed data on the space of bounded functions $\ell^\infty$. This approach is not only mathematically more natural, because pointwise evaluations exist; it also allows for our second major contribution, namely \textit{the development of the first simultaneous inference methods for partially observed functional data. In particular, we present simultaneous confidence bands (SCBs) for the mean parameter.} Simultaneous inference is, again, tricky because of measurability issues related to those known in stochastic process theory. To address this problem, our approach combines the classical idea of confidence interval interpolation  with recent advances in high-dimensional statistics (à la \cite{chernozhukov:chetverikov:kato:2015}). The resulting inference methods are a substantial improvement on SCBs even for fully observed functional time series, as they avoid the use of function-valued CLTs, while still leading to a  band diameter of $\mathcal{O}_{\mathbb{P}}(N^{-1/2})$. After completing this work, we became aware of the independently developed and parallel work of \cite{Sun:Cai:Hu:2026}, who also construct confidence bands for partially observed functional data. Compared to the present work their theoretical development requires stricter assumptions. They do not allow for dependence along the time series of missingness indicators and also implicitly require the data to be Lipschitz in $L^2$. Additionally they also provide no characterization of when constructing confidence bands with diameter  $\mathcal{O}_{\mathbb{P}}(N^{-1/2})$ is possible.

\textit{Our third contribution is the development of inference for non-stationary time series of partially observed data. Our probabilistic tools have use for stationarity testing, multiscale change point localization and detection of threshold exceedance.}
Our main focus is the last point, which is directly motivated by the monitoring of pollution data, and the risk of reaching dangerous pollutant concentrations. Mathematically, our approach relies on a generalization of multiscale results, recently advanced for scalar data by \cite{koehne:mies:2025}.

The remainder of this paper is structured as follows: In Section \ref{sec:math}, we present our statistical methodology and our main theoretical results. In Section \ref{sec:finite:sample}, we investigate finite sample properties of our new methods by virtue of simulation studies, comparison to other methods and an application to a time series of air pollution data. Mathematical proofs, technical details and further data analyses are collected in the Supplement. \\

\subsection{Notation} \label{sec:not}
Throughout this work, we denote by $\ell^{\infty}[0,1]$ the space of bounded functions on the unit interval, which is equipped with the supremum norm $|f|_\infty = \sup_{t \in [0,1]}|f(t)|$. The subspace of continuous functions is denoted by $\mathcal{C}[0,1]$. The space $L^p[0,1]$ refers to the measurable functions on $[0,1]$, equipped with the canonical $L^p$-norm denoted by $|f|_p$; notice that on this space functions that differ on a null set are identified with each other. 
We distinguish between $L^p$-norms (for functions) and $p$-norms (for random variables) to make our results clearer. For a real-valued random variable $Z$, we denote the $p$-norm by $\|Z\|_p$.

\section{Mathematical Theory} \label{sec:math}

In this section, we present our new statistical methodology and our main theoretical results. First, we consider the construction of simultaneous confidence bands, in Section \ref{sec:SCB}. This problem is of interest in many applications, and it also helps to familiarize the reader with our new inference approach for partially observed functional data. Second, in Section \ref{sec:loc}, we substantially extend our results to encompass local inference for non-stationary time series. Specifically, we validate the convergence of strongly weighted multiscale statistics that have use in stationarity testing, multiple change point detection and testing for the exceedance of thresholds. As we discuss along the way, our results extend current knowledge even for fully observed functional data.

\subsection{Simultaneous confidence bands for partially observed functional data}\label{sec:SCB}

Consider a sample of data functions $X_1(\cdot), \ldots, X_N(\cdot) \in \mathcal{C}[0,1]$, where each function $X_i(\cdot)$ is only observed on a random part of its domain denoted by $\mathcal{O}_i \subset [0,1]$. For mathematical convenience, we describe the observation scheme using indicators 
\begin{align*}
    O_i(\cdot):=\mathbb{I}\{\cdot \in \mathcal{O}_i\} \,\, \textnormal{and denote the $i$th observed function by} \,\, Y_i(\cdot):= O_i(\cdot)X_i(\cdot).
\end{align*}
Under stationarity and imposing certain moment conditions on the $X_i$ specified below, we can decompose each $X_i(\cdot)$ into a mean and a noise term as follows
\begin{align} \label{e:mod:1}
    X_i(t) = \mu(t) + \varepsilon_i(t), \quad 1 \le i \le N.
\end{align}
Here $\mu(\cdot):= \mathbb{E}X_i(\cdot)$ is the mean function, which is assumed to be  identical for all $i$ in this section, and  $\varepsilon_i(\cdot)= X_i(\cdot)-\mu(\cdot)$ is a centered noise. We also define $\pi(\cdot) =\mathbb{E}O_i(\cdot)$, which describes the probability of sampling an observation at location $t$, and is also independent of $i$ by our assumptions below.
The aim of statistical inference in this section is to construct a simultaneous confidence band (SCB) for the entire mean function $\mu(\cdot)$, at some asymptotically controlled coverage level $(1-\alpha)$. The approach should account for serial dependence along the functions $X_i(\cdot)$ and the observational pattern $O_i(\cdot)$. 

The approach described below is quite natural. It rests on the construction of a large number of simultaneous confidence intervals for the mean $\mu(t_j)$ over grid points $t_1,\ldots, t_p$. Subsequently the edges of these confidence intervals are interpolated, yielding a uniform confidence band across all $t$. Simple as this method may appear, the rigorous mathematical formulation and derivation turn out to be  quite challenging. Challenges include issues of measurability, as well as the formal proof that the confidence band is of rate-optimal width $\mathcal{O}_{\mathbb{P}}(N^{-1/2})$, implying that there are no  extra $\log$-factors resulting from the construction of the $p$ individual confidence intervals. 

As a first step to formalize our problem, we specify our data generating mechanism. We therefore invoke \textit{physical dependence}, a concept introduced by \cite{wu:2005}. In the following, we assume that there exists a representation
	\begin{equation}
    \label{eq:filter:formalism}
		(\varepsilon_{i} (t), O_i(t))  =   (G_1(t,\mathcal{F}_i),G_2(t,\mathcal{F}_i) )~,~~i=1, \ldots , N
	\end{equation}
	where $\mathcal{F}_i:=( \ldots , \eta_{i-1},\eta_i)$, $(\eta_i)_{i\in \Z} $ is a sequence of independent identically distributed
	random variables in some measurable space $\mathcal S$ 
	and $G: [0,1] \times {\mathcal S}^{\Z } \to \R^2$ denotes a measurable filter such that $G=(G_1(\cdot,\mathcal F_i), G_2(\cdot,\mathcal F_i))$ forms a bivariate sequence whose coordinates are independent.
Next,  denoting by $(\eta'_i)_{i\in \mathbb N } $  an independent copy of $(\eta_i)_{i\in \mathbb Z}$, we define $\mathcal{F}'_i:=( \ldots , \eta_{-1},\eta_0', \eta_{1}, \ldots ,  \eta_i)$ and therewith the version
	\begin{equation*}
		(\varepsilon_{i}' (t), O_i'(t))  =   (G_1(t,\mathcal{F}_i'),G_2(t,\mathcal{F}_i') )~,~~i\ge 1.
	\end{equation*}
We may now specify the mathematical assumptions for our subsequent analysis:
    \begin{enumerate}
        \item[(A)] (Moments) It holds that $\E[|\varepsilon_1(\cdot)|_\infty^4]< \infty$.
       \item[(B)] (Continuity) For some constant $\vartheta \in (0,1]$ the functions $\mu(\cdot)$ and $\pi(\cdot)$ are $\vartheta $-Hölder continuous. Moreover $\pi(t)>0$ for all $t \in [0,1]$. 
        \item[(C)] (Weak dependence of functions) Define the dependence measures
		\begin{align} \label{2.7}
			\delta_\varepsilon(i):=\sup_{t\in [0,1]}\|\varepsilon_i(t)-\varepsilon_i'(t)\|_{3}, \qquad \delta_O( i):=\sup_{t\in [0,1]}\|O_i(t)-O_i'(t)\|_{3}.
		\end{align}
		There exists a constant  $\chi\in(0,1)$ such that 
			$\delta_{\varepsilon}(i), \delta_{O}(i)=\mathcal{O}(\chi^i)~$ for $i\geq 1$ .
		   \item[(D)] (Weak dependence of increments) Define $H_1(i,t,s)=\varepsilon_i(t)-\varepsilon_i(s)$, $H_1'(i,t,s)=\varepsilon_i'(t)-\varepsilon_i'(s)$ as well as $H_2(i,t,s)=O_i(t)-O_i(s)$, $H_2'(i,t,s)=O_i'(t)-O_i'(s)$ and the dependence measure
        \begin{align*}
            \tau(i):=\sup_{s,t\in [0,1], j \in \{1,2\}}\frac{\|H_j(i,t,s)-H_j'(i,t,s)\|_2}{|t-s|^\vartheta }.
        \end{align*}
        It holds that $  \tau(i)=\mathcal{O}(\chi^i)$ for $i\geq 1$.
		\item[(E)]  (Long-run variances) We define the following long-run covariance 
		\begin{align} \label{e:lrv}
			\sigma_\varepsilon(t,s):=\sum_{k=-\infty}^\infty \C[\varepsilon_0(t),\varepsilon_k(s)]
		\end{align}
     and assume that $\sigma_\varepsilon^2(t):=\sigma_\varepsilon(t,t)$ is bounded away from $0$, uniformly over $t$. 
	\end{enumerate}

\begin{rem} (Discussion of Assumptions)
The construction of our data in \eqref{eq:filter:formalism} implies, in some appropriate sense, strict stationarity. The formulation is handy, because we do not require measurability of the entire process $\{O_i(t): t \in [0,1]\}$ in $\ell^\infty[0,1]$, which it typically fails to satisfy. This is a clear advantage over other concepts to describe weakly dependent time series, such as strong mixing.
Assumption (A) requires uniform moments of a small order for noise processes. (B) implies, in a quite weak sense, uniform continuity of the mean function of $X_i$ and the missingness scheme $O_i$. It is possible to weaken this condition further still, using more general Hölder spaces, but we do not deem this generalization worth pursuing  here; in the literature, parameter functions are often assumed to be quite smooth. (C) and (D) imply weak dependence of individual  evaluations and of increments, and such assumptions can be validated in  many time series models (see e.g. \cite{zhang:cheng:2018}, \cite{Dette:Wu:2026}). 
 Finally, 
assumption (E) requires  a uniformly positive long-run variance (notice that convergence of the sum follows from our weak dependence conditions). The uniform positivity is an artifact from our use of high-dimensional Gaussian approximations. There, assuming a positive variance, bounded away from $0$ across all components is standard, and it thus spills over into our theory. We point out that in high-dimensional theory it is well known that this uniform positivity assumption may be dropped at the expense of more laborious computations. Thus, we expect that it can also be dropped in our regime. However, we have not focused further on this possible extension. 
 
\end{rem}

\begin{ex} \label{rem:missing}(Missingness schemes) 
We briefly discuss possible missingness structures $(O_i)_{i \in \N}$ that satisfy the above assumptions. Existing literature focuses on independent missingness schemes, i.e. the case
where $O_i$ only depends on the current innovation $\eta_i$. A general setting of this type is described by \cite{Park:Chen:Simpson:2022}, and the examples discussed there  (see condition M2 in that paper) are covered by our theory. 

As an illustration, consider the following mathematical toy model, where points  inside an interval $[U^{(i)}_{(1)},U^{(i)}_{(2)}]$ are missing. Here, $U^{(i)}_{(1)},U^{(i)}_{(2)}$ are the order statistics of two independent uniform random variables $U^{(i)}_{1},U^{(i)}_{2}$ on $[0,1]$. In this case it is easy to show that the expected number of observations $\pi(t)$ is everywhere positive and Lipschitz in $t$. We also notice that, even in this simplistic model, the functions $O_i(\cdot)$ are not measurable in $\ell^\infty[0,1]$, underlining the need to evade such requirements in a general theory. 

While the above example gives a clean mathematical illustration, it seems too simple to capture any useful scenario, even in the abstract. To obtain a model that is still highly stylized but at least plausibly approximates reality, we turn to dependent time series. We note that the example we present is, to the best of our knowledge, not covered by existing theory and in particular not by  \cite{Sun:Cai:Hu:2026}, who allow the functional data itself to be dependent but still require $O_i$ to be an iid sequence.

When considering dependent missingness schemes, the simplest meaningful example is perhaps a double waiting time model. More concretely, suppose that data are observed from a sensor on a time interval $[0,N]$ where the data from the interval $[(i-1),i]$ correspond to the function $X_i$ (for an illustration see Figure \ref{Fig:sens}). For example, there might be $N$ days, with each interval of unit length representing data from a single day. It is natural to use an exponential waiting time model for determining the time until a sensor failure/repair. The waiting times are assumed to be homogeneous, and if sensors are operational at time $t \in [0,N]$, the distribution of the time until failure is given by a shifted exponential $t+Exp(\lambda_1)$. Conversely, if at time $t$ there is a sensor outage, the time until sensor repair is exponential with distribution $t+Exp(\lambda_2)$. In Figure \ref{Fig:sens}, failures are indicated by red shaded regions. 
   \begin{figure}[h] 
  \centering
  \includegraphics[width=1\textwidth]{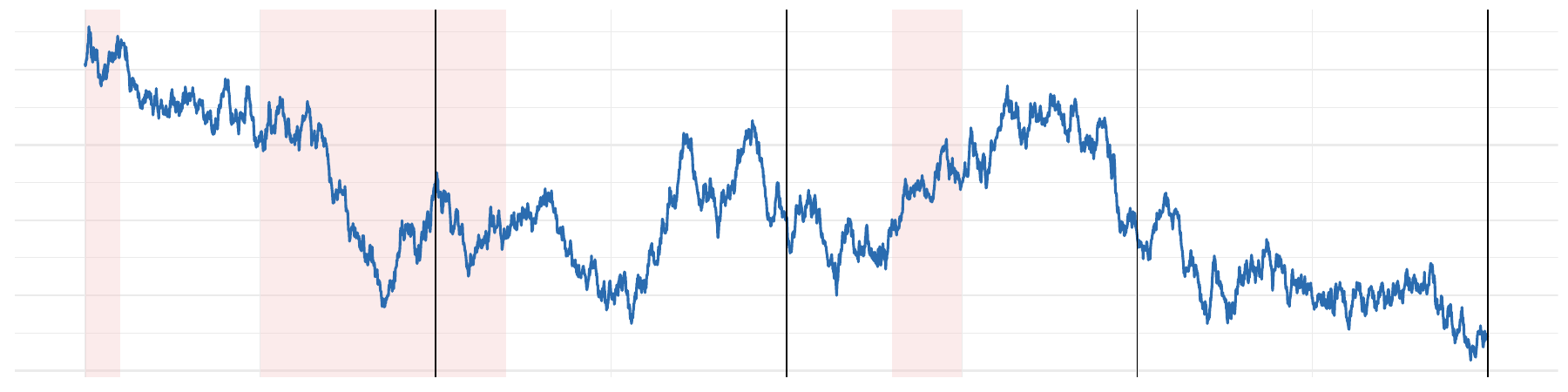}
  \caption{\label{Fig:sens} Stochastic process $X$ that is partitioned into $X_1,X_2,X_3,X_4$. Missing data are indicated by a red shaded region. The missingness scheme is generated by alternating exponential distributions that determine the duration of a detector in the operational and failure state. 
  }
  \label{fig:missingness}
\end{figure}
Carving up data along intervals gives us data functions $X_i(\cdot)$ and missingness variables $O_i(\cdot)$. We now explain how this data structure can be fitted into model \eqref{eq:filter:formalism}, generalizing it quite substantially in the process.
 Let $\eta_i=(E_{i,j})_{j}\in\R^\infty$ where $E_{i,j}\sim \text{Exp}(\lambda_{j})
 $ are independent exponentially distributed with parameters $\lambda_{j}$ bounded away from 0 and infinity. We define  sums of waiting times as
    \begin{align}
        S_{i,j}=\sum_{k=1}^jE_{i,2k}\quad, S_{i,0}=0, \qquad 
        T_{i,j}=\sum_{k=1}^jE_{i,2k+1}\quad, T_{i,0}=0
    \end{align}
    where $S_{i,j}$ describes the sum of the first $j$ waiting times on day $i$ when the previous day ended on $O_{i-1}(1)=1$ and $T_{i,j}$ describes them when it did not. Then, we may describe the missingness scheme by
    \begin{align}
    \label{eq:defin:waiting:missing}
        O_i(t)=\begin{cases}
            1 \quad  S_{i,2j}\leq t \leq S_{i,2j+1},& O_{i-1}(1)=1\\
            0 \quad  S_{i,2j+1}\leq t \leq S_{i,2j+2},& O_{i-1}(1)=1\\
            0 \quad  T_{i,2j}\leq t \leq T_{i,2j+1},& O_{i-1}(1)=0\\
            1 \quad  T_{i,2j+1}\leq t \leq T_{i,2j+2},& O_{i-1}(1)=0~.
        \end{cases}
    \end{align}
    The general model facilitates accounting for different waiting times, depending on how many failures occurred on each day and on how the previous day ended. To recover the simpler double-waiting-time-model we may choose, for any $i$ and $k$, $E_{i,2k}\sim \textnormal{Exp}(\lambda_1)$ when $2k$ is not divisible by $4$. Otherwise it should follow $\textnormal{Exp}(\lambda_2)$. Similarly $E_{i,2k+1} \sim \textnormal{Exp}(\lambda_2)$ when $2k$ is not divisible by $4$ and it should follow $\textnormal{Exp}(\lambda_1)$ otherwise.  This approach makes use of the fact that exponential distributions are memoryless.  The proof that the construction \eqref{eq:defin:waiting:missing} satisfies our main assumptions is non-trivial and given in the Supplement. 
\end{ex}

We now come to the development of our statistical approach to construct an SCB for the function $\mu(\cdot)$. As a first step, we define the empirical mean function
\begin{align} \label{e:def:pi}
   \hat \mu(t):=\frac{1}{\hat N(t)}\sum_{i=1}^N X_i(t)O_i(t), \quad \textnormal{where} \quad   \hat N(t):=\sum_{i=1}^NO_i(t).
\end{align}
$\hat \mu(t)$ is simply the pointwise empirical mean, using all data available at argument $t$; if no data are available $\hat \mu(t)$ is interpreted as $0$. Thus defined, $\hat \mu(\cdot)$ is a (not necessarily measurable) random element of $\ell^\infty[0,1]$. As an intermediate step, we aim at constructing confidence intervals for $\mu(t_j)$, for a large number of grid points $t_j, 1 \le j \le p$. In the following, we will assume that the grid points are placed such that $t_j < t_{j+1}$ for all $j$ and that
\[
\max_{1 \le j \le p-1}|t_{j+1} -t_j|\le C/p
\]
for some constant $C \ge 1$. This ensures that the mesh width is uniformly decreasing in $p$, and the standard choice is obviously the evenly spaced grid.  
Now, we define the maximum deviation
\begin{align} \label{e:hat:T}
\widehat{T}_N := \sqrt{N} \max_{1 \le j \le p} |\hat \mu(t_j)-\mu(t_j)|.
\end{align}
In our proofs, we demonstrate that $\widehat{T}_N $ can be approximated by the maximum of $p$ Gaussians as long as $p$ is polynomially growing in $N$. We notice that by a standard technical argument it is also possible to divide by the long-run variance for each $t_j$ and thus get variance-adapted confidence regions in our subsequent results. We will discuss these adaptations in our experiments, but omit their further discussion here for the sake of clarity.

In itself, the proof is non-trivial, using the high-dimensional approximations by \cite{zhang:cheng:2018}. The really interesting result is, however, that as $N, p \to \infty$ the statistic $\widehat{T}_N$ converges to the maximum of a tight Gaussian process $\{W(t): t \in [0,1]\}$. In particular, the limit is independent of the specific choice of the grid points $t_j$ or their number $p$, as long as $p$ is not growing too fast relative to $N$. This is a pleasant result, because it entails that all of the individual $p$ confidence intervals $ \hat C^\pm(t_j)$, constructed below, are of width $\mathcal{O}(N^{-1/2})$.  
It is also somewhat surprising, because we can easily construct examples where (under the conditions of this paper) we have divergence of the uniform distance without grid
\[
\widehat{T}_N^\infty:=\sqrt{N} \sup_{t \in [0,1]} |\hat \mu(t)-\mu(t)| \overset{d}{\to} \infty
\]
and in particular weak converge $\widehat{T}_N^\infty \overset{d}{\to}|W(\cdot)|_\infty$ does not hold. In this sense, our below approach of using confidence intervals over grids and interpolating them is genuinely weaker than trying to show convergence of the uniform distance $\widehat{T}_N^\infty$. A general discussion of this fact is given in Remark \ref{rem:SCBs}. 

The limiting process $W(\cdot)$ is a centered Gaussian and characterized by the covariance structure 
       \begin{align} \label{e:def:cov:W}
        \C[W(t),W(s)]=(\pi(t)\pi(s))^{-1}\sigma_{\textbf{Y}}(t,s), \qquad \sigma_{\textbf{Y}}(t,s):= \sum_{k=-\infty}^\infty \C[Y_0(t),Y_k(s)].          
       \end{align}
As a consequence of the Gaussian approximation $\widehat{T}_N \overset{d}{\to} |W(\cdot)|_\infty$, and denoting the $(1-\alpha)$-quantile of $|W(\cdot)|_\infty$ by $q_{1-\alpha}$, we obtain the confidence intervals 
    \begin{align} \label{e:int:sing}
        \hat C^\pm(t_j)&:= \hat \mu(t_j)\pm q_{1-\alpha}N^{-1/2}.
    \end{align}
    By construction, these intervals asymptotically hold a level of $(1-\alpha)$ simultaneously for the parameters $\mu(t_1),...,\mu(t_p)$. 
    Interpolating the edges of these confidence  intervals then yields a  confidence interval for any $t \in [0,1]$ as
   \begin{align} \label{e:int:SCB}
      \hat C^\pm(t):=\frac{t_{j+1}-t}{t_{j+1}-t_j}\hat C^\pm(t_j)+\frac{t-t_j}{t_{j+1}-t_j}\hat C^{\pm}(t_{j+1}), \quad t_{j}\leq t \leq t_{j+1}.
      \end{align}
      Here, if $0\leq t<t_1$ or $t_p<t\leq 1$ we set $\hat C^\pm(t)=\hat C^\pm(t_j),$ for $ j=1$ and $j=p$, respectively.
The idea behind interpolation is straightforward: If $\mu(t_j) \in [\hat C^-(t_j), \hat C^+(t_j)]$ and $\mu(t_{j+1}) \in [\hat C^-(t_{j+1}), \hat C^+(t_{j+1})]$ it should be the case that $\mu(t)$ lies within the interpolating corridor, supposing that $\mu(\cdot)$ is sufficiently smooth.
\begin{figure}[h]
  \centering

  \begin{subfigure}{0.43\textwidth}
    \centering
    \includegraphics[width=\linewidth]{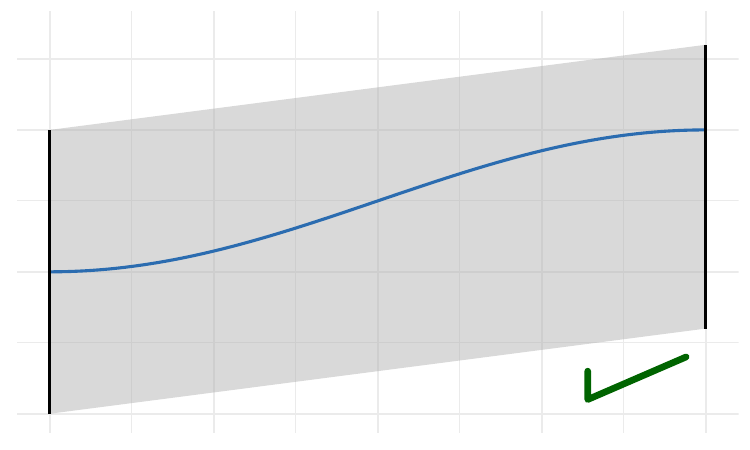}
  \end{subfigure}
  \hfill
  \begin{subfigure}{0.43\textwidth}
    \centering
    \includegraphics[width=\linewidth]{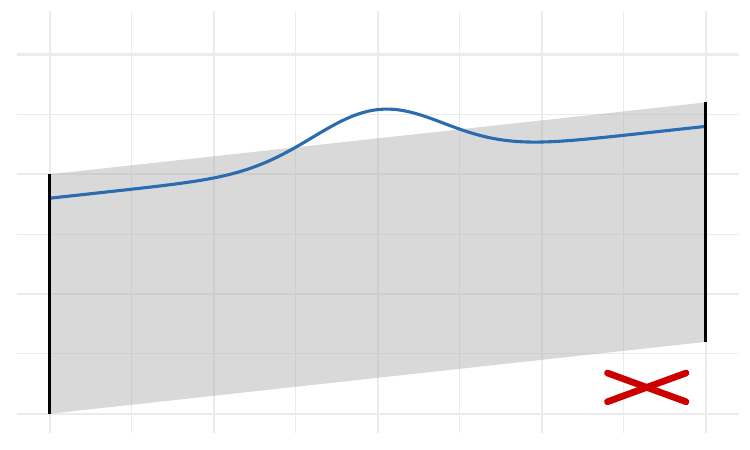}
  \end{subfigure}

  \caption{ \label{fig:escape} Two confidence intervals (vertical black), capturing the mean function (blue) pointwise. The interpolated confidence corridor (shaded gray) captures the mean function entirely (left) or fails to capture it (right).}
  \label{fig:ci_corridors}
\end{figure}
An additional argument is however needed to ensure that the mean function does not "escape" from the interpolating corridors but stays within them (see Figure \ref{fig:escape} for a visualization). The reason that such escapes (asymptotically) do not happen is, roughly speaking, that a mean function in $t_j$ is only with very low probability hitting the edge of a confidence interval $\hat C^\pm(t_j)$, but typically lies rather close to the middle. Hence, it manages to escape the confidence corridor between intervals only with negligible probability. The mathematical argument for this behavior relies on an anti-concentration result for $p$-dimensional Gaussian distributions. We can now formulate our first main result that guarantees SCBs with asymptotic coverage $(1-\alpha)$.

    \begin{theo}
    \label{theo:confidence:bands}
       Suppose that Assumptions (A)-(E) hold and that $p \sim  N^{C^*} $ for some $C^*>\frac{1}{2\vartheta}$. Let $W=\{W(t)\}_{t \in [0,1]}$ be the centered Gaussian process from eq. \eqref{e:def:cov:W}.
       Then almost surely $W\in \mathcal{C}[0,1]$. Moreover, denoting as before the $(1-\alpha)$-quantile of $|W(\cdot)|_\infty$ by $q_{1-\alpha}$, we have that 
        $\{\hat C^\pm(t)\}_{t \in [0,1]}$ (defined in \eqref{e:int:sing} and \eqref{e:int:SCB}) is an asymptotic $(1-\alpha)$ SCB  for $\mu(\cdot)$.
        
    \end{theo}
In the above theorem, the condition $C^*>\frac{1}{2\vartheta}$ means that more grid points are required if $\mu(\cdot)$ is less smooth (smaller $\vartheta$). This is a natural requirement, since a rougher $\mu(\cdot)$ can more easily escape interpolating confidence corridors if they are too long. We notice that in FDA parameter functions are often assumed to be fairly smooth. 
The quantile $q_{1-\alpha}$ is practically unknown, but can be bootstrapped, using a Gaussian process $\widehat W(\cdot)$ with  $\widehat W\sim \mathcal{N}(0,\widehat \C)$. More precisely, we may simply simulate the $(1-\alpha)$ -quantile from $|\widehat W(\cdot)|_\infty$, under the constraint that
        \[
            \widehat W\sim \mathcal{N}(0,\widehat \C), \quad  \sup_{\substack{s,t}}|\widehat \C(s,t)-\C[W(t),W(s)]|=O(N^{-\gamma}), \quad \gamma>0,
        \]   
        which yields a consistent estimate for $q_{1-\alpha}$. Suitable covariance estimators $\widehat \C$ are available in the literature (e.g. \cite{Li:Chen:Wang:Wu:2024}). We will provide an explicit construction in the finite sample section.

\begin{rem}  \label{rem:SCBs}(Characterizing when SCBs are possible)

Our approach is presented for partially observed functional data, but even for the fully observed case it holds important benefits. Current SCBs for fully observed functional time series are typically justified by showing that the process $\sqrt{N}(\hat \mu(\cdot)-\mu(\cdot))$ converges weakly, then applying a continuous mapping theorem to conclude convergence of the functional $\widehat{T}_N^\infty:=\sqrt{N}|\hat \mu(\cdot)-\mu(\cdot)|_\infty$. In contrast, our approach only requires convergence of the  discretizations $\widehat{T}_N=\sqrt{N}\max_{1 \leq j \leq p}|\hat \mu(t_j)-\mu(t_j)|$ for a grid of polynomial size. It is important to point out that convergence of the discretized maximum $\widehat{T}_N$ is strictly weaker than convergence of the full supremum $\widehat{T}_N^\infty$, and this in turn is strictly weaker than assuming process convergence. In the Supplement, we construct an example where $\widehat{T}_N$ converges while $\widehat{T}_N^\infty$ does not.
We also derive both necessary and sufficient criteria for the validity of SCBs based on discretized maximum statistics (with interpolating corridors). To the best of our knowledge, this provides the first characterization of this type, making a significant contribution to the literature. The gains are particularly obvious in the dependent setting. Here, previous approaches invariably involve a trade-off between smoothness and moment conditions to establish process-level convergence. Remarkably, such a trade-off is entirely absent from our Theorem \ref{theo:confidence:bands}. The only thing that is needed is a (kind of) smoothness of the limiting Gaussian process, which is very weak.

Finally, and returning to the subject of partially observed functions, we highlight that under our Assumptions (A) to (E), it is possible to construct counterexamples to convergence of $\widehat{T}_N^\infty$. This means that discretization is not just an artifact of our proof technique, but rather that it is essential to obtain valid SCBs in the partially observed case.
\end{rem}

\begin{rem}(Related Approaches)

    The idea of interpolating uniform confidence intervals to SCBs is a classical approach, even though the combination with high-dimensional approximations seems to be new in the FDA literature (and in particular to partially observed FDA). 
    High-dimensional Gaussian approximations have, of course, been used in the context of fully observed FDA, with a recent instance being the work of \cite{Dette:Wu:2026}. Here, a maximum deviation is approximated for subsequent simultaneous inference. More broadly, high-dimensional Gaussian approximations have been used in the context of non-parametric estimation problems, as in the classical work by \cite{Chernozhukov:Chetverikov:Kato:2014} for probability densities. These authors construct uniform confidence bands  based on VC class properties. This task is however different from our problem for partially observed FDA, where data functions $Y_i$ are technically members of a class with infinite VC dimension.
    
\end{rem}

\subsection{Local inference via multiscale statistics}\label{sec:loc}

In the previous section, we have focused on inference for the fixed mean function $\mu(\cdot)$ of a stationary and partially observed functional time series; see model \eqref{e:mod:1}. We now develop inference tools for the mean function in a non-stationary setting where
\begin{align} \label{e:mod:2}
    X_i(t) = \mu_i(t) + \varepsilon_i(t), \quad t \in [0,1], \,\,\,1 \le i \le N.
\end{align}
As before, $\varepsilon_i(\cdot)$ denotes a centered noise process and data are only partially observed on a domain characterized by $O_i(\cdot)$ - both generated by the mechanism \eqref{eq:filter:formalism}. This time, however, the mean function $\mu_i(\cdot)$ may depend on $i$, making the time series potentially non-stationary. In this setting a variety of inference tasks exist, depending partly on underlying assumptions placed on the evolution of the mean $\mu_i(\cdot)$ over time. We list three inference tasks that are relevant in the functional model \eqref{e:mod:2} and for which we develop methodology. \\
 \hspace*{1.5em}\textit{I) Stationarity testing:}  The most basic task in model \eqref{e:mod:2} is testing for mean stationarity, i.e. testing the hypothesis pair
\begin{align} \label{e:H0:stat}
H_0^{stat}: \mu_i(\cdot) =\mu_j(\cdot),\,\, \forall i,j \qquad \textnormal{vs.} \qquad    H_1^{stat}: \mu_i(\cdot) \neq\mu_j(\cdot) \,\,\textnormal{for some} \,\,i,j.  
\end{align}
   \hspace*{1.5em}\textit{II) Change point inference:} A popular way to model mean evolution is using a piecewise constant model. Here, there exist change points $c_0:=0<c_1< c_2< \ldots <c_K:=N$, with constant mean between changes, i.e.
    \[
    \mu_i(\cdot) := \sum_{k =1}^K\mu^{(k)}(\cdot) \, \mathbb{I}\{c_{k-1} < i \le c_k\}
    \]
     $\mu^{(k)}(\cdot)$ refers to a mean function for  $1 \le k \le K$.
    The task is then to estimate the total number of changes and create (simultaneous) discrete confidence intervals for all changes $c_k$.\\
    \hspace*{1.5em}
 \textit{III) Threshold exceedance:} Rather than using a number of abrupt changes in the mean parameter, one can model a continuous evolution over time.
For this purpose let  $\nu: [0,1] \times [0,1] \to \mathbb{R}$ be a continuous function with
    \[
    \mu_i(t) := \nu(t,i/N).
    \]
    A practically relevant task is then testing for threshold exceedance of $\nu$. Here, the user determines a (scientifically meaningful) threshold $\Delta>0$ and wants to test the hypothesis pair
    \begin{align} \label{e:threshold}
          H_0^\Delta: \sup_{(t,x) \in [0,1]} \nu(t,x) \le \Delta \qquad \textnormal{vs} \qquad     H_1^\Delta: \sup_{(t,x) \in [0,1]} \nu(t,x) > \Delta.  
    \end{align}

\medskip
All the above inference tasks can be addressed in  essentially the same way: By using multiscale statistics. Below, we will introduce two related statistics, one for tasks $I)$ and $II)$, and a related one for task $III)$.
The last one, for threshold exceedance, is particularly interesting for our empirical analysis of pollution  data.  There, we  test for air pollution levels $\nu$ crossing a threshold $\Delta$ that entails negative health outcomes. 

We now give details on the multiscale method. Multiscale statistics stem from the study of temporal and spatial processes, to analyze phenomena that may occur at different scales or resolutions. The fundamental idea is to test a model hypothesis, such as \eqref{e:threshold}, using many subsamples of the data - some subsamples that are small and localized, to detect short and strong violations of $H_0$, and some subsamples large and spread out, to detect global and more subtle  violations of $H_0$. For time series, we consider all subsamples $\{X_i, X_{i+1},\cdots, X_{j}\}$, $1 \le i \le j\le N$ of consecutive observations. 
The combination of the test statistics for all samples under consideration then involves a weighting scheme to assign relative importance to the different scales. Typically smaller scales are slightly discounted to avoid them dominating the statistic, and in our case this is expressed by a logarithmic discounting factor. Now, to mathematically define our statistic, we first define the empirical probability of making an observation at $t$, $\hat \pi_N(t):=\hat N(t)/N$ and second, the partial sums 
\[
    U_i^j(t):=\sum_{m=i+1}^jX_m(t)O_m(t)~.
\]
Then, we define the two multiscale statistics as functions of an argument $t$
\begin{align}  \label{e:hatM}
\widehat M_N(t):=& \max_{1 \le l \leq k \leq N-l} \big\{ m(t,k,l) \big\}, \quad m(t,k,l):= \frac{U_{k-l}^k(t)-U_k^{k+l}(t)}{\hat \pi_N(t)\sqrt{l}\log^{2}(eN/l)}\\
       \widehat R_N(t) := & \max_{1 \leq l< k\leq N} \big\{ r(t,k,l) \big\}, \quad\quad \,\,r(t,k,l):=
       \frac{U_l^k(t)-(k-l)\Delta}{\hat \pi_N(t)\sqrt{k-l}\log^{2}(eN/(k-l))} \label{e:hatR}
\end{align}
As before, we interpret these statistics as $0$ if $\hat \pi_N(t)=0$ (no points sampled in $t$).
Both $\widehat M_N(t)$ and $  \widehat R_N(t)$ are essentially a combination of many test statistics, across different locations and sample sizes. Thus, they automatically adapt to signals of different lengths and intensity, with a small illustration (without missing data)  given in Figure \ref{fig:multi}.

\begin{figure}[h]
  \centering

  \begin{subfigure}{0.46\textwidth}
    \centering
    \includegraphics[width=\linewidth]{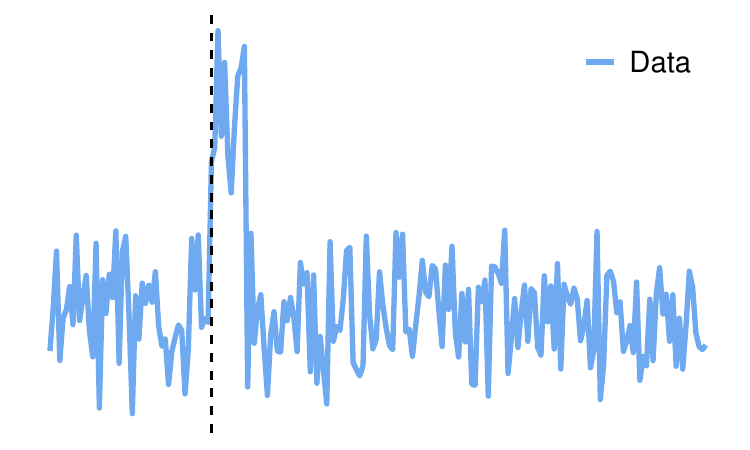}
  \end{subfigure}
  \hfill
  \begin{subfigure}{0.46\textwidth}
    \centering
    \includegraphics[width=\linewidth]{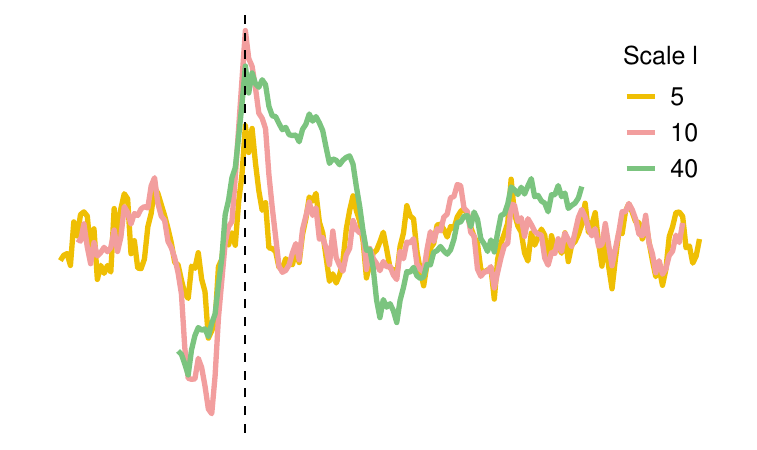}
  \end{subfigure}

  \caption{ \label{fig:multi} Left: Synthetic data $X_i(t)$ at $t=1/2$ (fully observed case) for $N=200$, with a short, spiked signal of length $10$ beginning at time $50$. Right: Calculation of  $m(t,k,l)$ for $k=50, t=1/2$ and different values of the scale parameter $l=5,10,40$. We see that the choice $l=10$ (true signal length) yields the maximum value for the statistic.}
\end{figure}

The statistic $|\widehat M_N(\cdot)|_\infty$ has been previously used for fully observed functional time series in \cite{kutta:dette:wang:2025}, and we refer to that paper for algorithms on multiscale change point detection.  Analogously to our discussion in Remark \ref{rem:SCBs}, we notice that the approach presented here relies on weaker assumptions than standard FDA, such as used in \cite{kutta:dette:wang:2025}, who require invariance principles on a function space, rather than our more parsimonious Gaussian approximation for $|\widehat M_N(\cdot)|_\infty$. Therefore, again, our theory represents an improvement even for the case of fully observed FDA.

We now discuss the main assumptions of this section. In the last section, we imposed (A)-(E). Here, we strengthen the first three assumptions, because our multiscale approach requires more regularity. Recall also the filtering scheme \eqref{eq:filter:formalism} used to generate noise and missingness processes.
    	\begin{enumerate}
	    \item[(A')] (Moments) It holds that
        \[
            \E[\exp(|\varepsilon_1(\cdot)|_\infty^2)]< \infty.
        \]
        \item[(B')] (Continuity)
        For some constant $\vartheta \in (0,1]$ the functions $\mu_i(\cdot), 1 \le i \le N$ and $\pi(\cdot)$ are uniformly $\vartheta $-Hölder continuous. Moreover $\pi(t)>0$ for all $t \in [0,1]$. 
        \item[(C')] (Weak dependence of functions) Define the  dependence measures
		\begin{align} \label{3.7}
			\delta_{q}( i)&=\max\Big\{\sup_{t\in [0,1]}\|\varepsilon_i(t)-\varepsilon_i'(t)\|_{q},\sup_{t\in [0,1]}\|O_i(t)-O_i'(t)\|_{q}\Big\}, 
		\end{align}
        and define $
            \Theta_q=\sum_{i=1}^\infty \delta_q(i)$.
        Now, for some $\chi \in (0,1)$ it holds that
		\begin{align}\label{3.8}
			\sup_{q \geq 2}q^{-1/2}\Theta_q<\infty, \quad \delta_3(i)=\mathcal{O}(\chi^i)~.
		\end{align}
	\end{enumerate}
Assumptions (A') and (C') are strengthened versions of (A) and (C), required such that exponential concentration results hold. (B') is a uniform version of (B), enforcing Hölder continuity across all (possibly different) means $\mu_i(\cdot)$.

For the statement of our main result, one more mathematical prerequisite is required: The Brownian motion corresponding to a Gaussian random function $W(\cdot)$. Recall that in the last section we discussed the Gaussian random function $W(\cdot)$, which takes values in $\mathcal{C}[0,1]$. The corresponding Brownian motion $\{B_W(x): 0 \le x \le 1\}$ is a centered Gaussian process that is a.s. continuous, and with its increments independent and satisfying $B_W(y)-B_W(x) \overset{d}{=} \sqrt{y-x}\cdot  W$. Notice that, for any $x \in [0,1]$, $B_W(x)$ is still a random function in $\mathcal{C}[0,1]$ and thus we may also interpret it as a real-valued function with two arguments $B_W(x,t)$.
Abstractly, $B_W$ is a Brownian motion on a separable Banach space (of continuous functions) and it can be shown to behave in most regards like its finite-dimensional counterpart; in particular it satisfies a Lévy modulus of continuity result, which implies strong uniform smoothness properties.  

The multiscale statistics of this section then behave (in the absence of a signal) like some modulus of continuity of $B_W$. 
We now specify these moduli of continuity, which are needed for our limiting results: Therefore, let  $H:[0,1]^2 \to \mathbb{R}$ be a function and define
\begin{align*}
    \Phi(H):= & \max_{\substack{0 \le x \le 1\\
                0 < y \le \min(x,1-x)}}
\frac{|H(x+y, \cdot)-2H(x,\cdot)+H(x-y,\cdot)|_\infty}{\sqrt{y}\log^{2}(e/y)},\\
    \Psi(H):= & \max_{0\leq x<y\leq 1}\max_{0 \leq t \leq 1}  \frac{H(x, t)-H(y, t)}{\sqrt{y-x}\log^2(e/(y-x))}~.
\end{align*}
For the following theorem, denote by $q_{1-\alpha}^\Phi$, $q_{1-\alpha}^\Psi$ the $(1-\alpha)$ quantiles of the distribution $ \Phi(B_W)$ and $\Psi(B_W)$ respectively.

\begin{theo}
    \label{theo:local:inference}
    Assume that (A')-(C') as well as (D) and (E) hold. Further suppose that $p \sim N^{C^*}$ for some $C^*>\frac{1}{2\vartheta}$. Then, there exists an $\alpha^* \in (0,1)$ depending only on the generating mechanism \eqref{eq:filter:formalism}, such that for any $\alpha\le \alpha^*$ it holds that:
    \begin{itemize}
        \item[i)] Under the null hypothesis of mean stationarity (eq. \eqref{e:H0:stat}) it follows that
        \[
       \lim_{N \to \infty} \mathbb{P}\big(\max_{1 \le j \le p}|\widehat M_N(t_j)|>q_{1-\alpha}^\Phi\big)=\alpha.
        \]
        \item[ii)] Under the null hypothesis of no threshold exceedance (eq. \eqref{e:threshold}) it follows that
        \[
       \limsup_{N \to \infty} \mathbb{P}\big(\max_{1 \le j \le p}\widehat R_N(t_j)>q_{1-\alpha}^\Psi\big) \leq \alpha.
        \]
        In the above formula we may replace "$\limsup_{N \to \infty}$" by "$\lim_{N \to \infty}$" and "$\le$" by "$=$", if $\nu(\cdot, \cdot) \equiv \Delta$.
        
    \end{itemize}
\end{theo}
    \begin{rem}(Details on Theorem \ref{theo:local:inference}) \label{rem:details}\\
    $ {}\quad {}$ \indent \textit{I) Proof:} The proof of Theorem \ref{theo:local:inference} is quite challenging, much more so than our previous results on confidence bands. 
    
    The Gaussian approximation for a functional multiscale statistic
    is highly non-standard, proceeding over many arguments $t_j$, $k$ and scales ($k-l$ and $l$ respectively).     
    Asymptotically, not all scales actually behave like Gaussians, and as a consequence we cannot show the standard weak convergence 
    \[
    \max_j|\widehat M_N(t_j)| \overset{d}{\to} \Phi(B_W).
    \]
    Indeed, this weak convergence may fail to hold, when the small scales' contribution to the statistic does not behave like the small scales of the limiting Gaussian (see \cite{koehne:mies:2025}). However, what we can show is that the maximum over very small scales yields an (asymptotically) bounded random variable. Accordingly, large enough quantiles ($(1-\alpha)$ with $ \alpha \le \alpha^*$), are determined by large scale contributions, and these have asymptotically Gaussian fluctuations. Gaussian approximation for larger scales is still non-trivial, because it requires us to extend results by \cite{zhang:cheng:2018}, making certain bounds more quantitative and allowing for asymptotically vanishing variances. \\ 
    $ {}\quad {}$ \indent \textit{II) Threshold exceedance:} The test for threshold exceedance approximates exact level $\alpha$ in the (practically somewhat unrealistic) scenario of $\nu(\cdot, \cdot) \equiv \Delta$; otherwise it is conservative. However, exact level approximation is possible even when $\nu(t,x)=\Delta$ just on a very small subset of indices. In this case, quantiles have to be adjusted using an estimator of the set $\{(t,x): \nu(t,x)=\Delta\}$, and we discuss such an estimator in Section \ref{sec:app:adaption} of the Supplement. \\
    $ {}\quad {}$ \indent \textit{III) Small scales:} The statistics $\widehat M_N(t)$ and $\widehat R_N(t)$ are divided by a log-factor of the form $\log^2(eN/\textnormal{scale})$, which suppresses small scales in favor of large scales. There exists, however, some flexibility regarding the normalizing factor, and different applications may benefit from modifications. For instance,
    if it is suspected that signals will occur mostly at small scales, one may adjust the weights to put additional emphasis on small scales, using $\log^2(C_0 + N/\textnormal{scale})$ for some sufficiently large constant $C_0$. Larger values of $C_0$ increase power at small scales and reduce power at larger scales. 
    All results from this section remain true ceteris paribus for this weight and any $C_0>1$, and we recommend the use of $C_0=20$ in applications.
    
    \end{rem}
Theorem \ref{theo:local:inference} is central to controlling the behavior of the multiscale statistics under a null hypothesis. We now want to briefly discuss the behavior, and specific strengths, of the multiscale approach when detecting a signal. To save space, we do not discuss the problem of change point inference, which involves some mathematical intricacies in a rigorous formulation. We refer the interested reader to \cite{kutta:dette:wang:2025} for details. Here, we consider local alternatives to the hypothesis of stationarity \eqref{e:H0:stat} and no threshold exceedance \eqref{e:threshold}. In the following, the functions $\mu, \mu^*$ are assumed to be Hölder continuous with some index $\vartheta \in (0,1]$.  For stationarity testing, we consider local alternatives of an epidemic change point
\begin{align} \label{e:alt:1}
\mu_i(t) = \mu(t) + \mu^*(t) \,\mathbb{I}\{\lfloor Nx_0 \rfloor \le i \le \lfloor N(x_0+a_N) \rfloor  \}, \quad |\mu^*|_\infty\ge b_N,  
\end{align}
where $x_0 \in (0,1)$ locates the change, $N a_N$ (with $a_N>0$) describes the length of the segment where the mean is different and $b_N$ (with $b_N>0$) the size of the change.  
Having more power against shorter segments and smaller signal magnitudes is obviously preferable. 

For the null hypothesis of no threshold exceedance, we consider as an alternative a gradual deformation
\begin{align} \label{e:alt:2}
\mu_i(t) = \mu(t) \Big(1+ (b_N/\Delta)\cdot d(|i/N-x_0|/a_N) \Big),
\end{align}
where $x_0 \in (0,1)$ and $d: [0,\infty) \to \mathbb{R}$  is a monotonically decreasing deformation function that satisfies $d(0)=1$ and $d(1)=0$. We assume that $\sup_t \mu(t)=\Delta$ such that $b_N$ is again the size of the signal.  The size of $a_N N $ is again some proxy for the minimal length of the signal - however it needs to be considered together with the smoothness of $d$ in $x=0$. We quantify behavior close to the origin by positing that there exists some $\kappa>0$ such that
\[
\lim_{x \downarrow 0} \frac{d(0)-d(x)}{x^\kappa}=1.
\]
The next theorem investigates consistency against local alternatives. Details on the interpretation can be found in Remark \ref{rem:loc:cons} in the Supplement.
\begin{theo} \label{thm:loc:cons}
    Suppose that $p \sim N^{C^*}$ for some $C^*>\frac{1}{2\vartheta}$, and impose Assumptions (A')-(C') as well as (D)-(E). 
    \begin{itemize}
        \item[i)] Assume that $\sup_N a_N <1$. Then, if model \eqref{e:alt:1} holds true and if $a_N, b_N$ satisfy  $b_N\sqrt{a_NN}\log^{-2}(e/a_N)\to \infty$,
         we have $\max_{1 \le j \le p}|\widehat M_N(t_j)| \overset{P}{\to} \infty$.
        \item[ii)] Assume that $\sup_N a_Nb_N^{1/\kappa} <1$. Then, if model \eqref{e:alt:2} holds true and if $a_N, b_N$ satisfy      $b_N^{1+1/(2\kappa)}\sqrt{Na_N}\log^{-2}(e/(a_Nb_N^{1/\kappa})) \to \infty$,        we have 
        $\max_{1 \le j \le p}\widehat R_N(t_j)\overset{P}{\to} \infty$.
        
    \end{itemize}
    
\end{theo}

\section{Finite Sample Performance}
\label{sec:finite:sample}

We explore the numerical performance of our new methods using simulations and a data analysis.

\subsection{Confidence Bands}

We first discuss the SCBs, developed in Section \ref{sec:SCB}. We begin with i.i.d. data and then move to dependent time series. 
For independent functions, we discuss two settings, namely
\begin{align}
\label{e:Xind} (I) \quad X_i(\cdot)&=B_i(\cdot), \quad \textnormal{and}\quad  (II) \quad   X_i(\cdot)=\sum_{k=1}^{10} f_k(\cdot) T_{ik}
\end{align}
where $(B_i(\cdot))_i$ are i.i.d. Brownian bridges on $[0,1]$, $(T_{ik})_{i,k}$ an array of i.i.d. $t$-distributions with 5 degrees of freedom and $(f_i(\cdot))_{i=1,...,10}$ is the B-spline basis as implemented in the \texttt{R} package "fda" \cite{fdaR}. 
In the dependent case we follow \cite{aue:rice:sonmez:2018} and consider first-order functional autoregressive processes
\begin{align}\label{e:Xdep}
   (I) \quad  X_i(\cdot)&=[\Psi X_{i-1}](\cdot)+\zeta_{i,l}(\cdot),\quad \textnormal{and}\quad  
 (II) \quad  X_i(\cdot)=[\Psi X_{i-1}](\cdot)+\zeta_{i,h}(\cdot)
\end{align}
where the innovations are given by
\begin{align} 
    \zeta_{i,l}(\cdot)=\sum_{k=1}^{21}N_{i,k}v_k(\cdot)k^{-1}, \quad \textnormal{and}\quad 
    \zeta_{i,h}(\cdot)=\sum_{k=1}^{21}T_{i,k}v_k(\cdot)k^{-1}.
\end{align}
In the innovations $N_{i,k}$ and $T_{i,k}$ are given by i.i.d. standard normal and $t$-distributions with 5 degrees of freedom, respectively. $v_1(\cdot),...,v_{21}(\cdot)$ are the Fourier basis functions on the interval $[0,1]$. The operator $\Psi$ is characterized by a random matrix $\Psi_0$ with $\Psi_0 \in \mathbb{R}^{21\times 21} $ whose entries are given by independent centered normals with variances $((2ij)^{-1})_{1 \leq i,j \leq 21}$. $\Psi$ then acts on the Fourier basis elements as follows $\Psi[v_i](\cdot) := \sum_{j=1}^{21} v_j(\cdot) (\Psi_0)_{j,i}$. 
Note that in both above settings (dependent and independent), the created functions are fairly rough and that in the second setting, dependence is relatively strong. In both settings, case (I) corresponds to light tailed and case (II) to heavier-tailed data. 

Next, we specify the missingness mechanisms. Again we consider an independent and a dependent scenario. For the independent case, we follow Example 1 from \cite{Park:Chen:Simpson:2022} with $\beta_{i,j}\overset{i.i.d.}{\sim}\mathrm{Beta}(0.2, 0.2)$, and set
\begin{align}
\label{eq:defin:interval:missing}
    O_i(t)=1-\mathbb{I}\{\min(\beta_{i,1},\beta_{i,2})\leq t \leq \max(\beta_{i,1},\beta_{i,2})\}~.
\end{align}
This describes missingness on a random interval. For a dependent missingness scheme, we use the double exponential model described in Remark \ref{rem:missing}, and more specifically in and after eq. \eqref{eq:defin:waiting:missing} with
 $\lambda_1=3$ and  $\lambda_2=1.5$ . We display the associated fraction of observed data at time $t$ (i.e. $\pi(t)$) in Figure \ref{fig:missingness2}. Notice that in both cases roughly half of all observations are on average missing at any location.

\begin{figure}[ht] 
  \centering
  \includegraphics[width=0.7\textwidth]{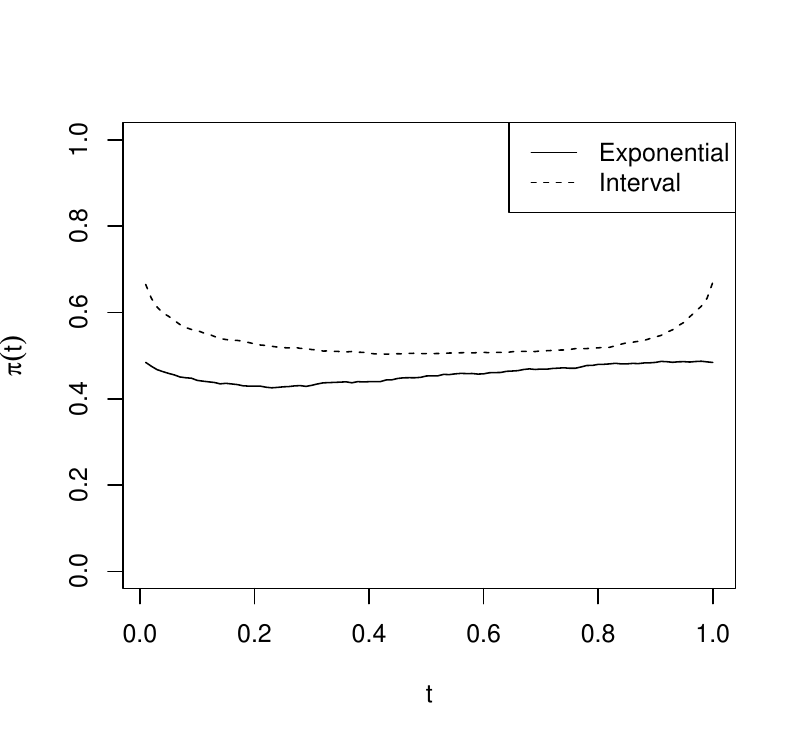}
  \vspace{-0.5cm}
  \caption{ Proportion of observed samples at location $t \in [0,1]$ for the missingness schemes \eqref{eq:defin:interval:missing} ("interval") and \eqref{eq:defin:waiting:missing} ("exponential"). Curves are approximated using $10^4$ runs. 
  }
  \label{fig:missingness2}\vspace{-0.5cm}
\end{figure}

As we have pointed out after eq. \eqref{e:hat:T} it is often useful to normalize the width of the confidence bands by the long-run variance, and our below results refer to such variance-adapted intervals. 
To that end we will use the long-run variance estimator $\hat \sigma_{\textbf{Y}}(s,t)$ from \cite{rice:shang:2017} with the Bartlett kernel, to approximate the pointwise long-run (co)variances of $Y_i(t)=X_i(t) O_i(t)$. In particular we use the therein proposed method to select the bandwidth parameter for estimation. We thus obtain the estimate 
\[
    \widehat \C(s,t)=\frac{N^2}{\hat N(t)\hat N(s)}\hat \sigma_{\textbf{Y}}(s,t)
\]
for the long-run covariance structure of the empirical mean function $\hat \mu(\cdot)$. We will now fix the number of confidence intervals in \eqref{e:int:sing} at $p=N$, with equidistant grid points and thus obtain the intervals 
\[
    \hat C^\pm(i/N)=\hat \mu(i/N)\pm\hat N(i/N)^{-1/2}\widehat{\C}^{1/2}(i/N,i/N)\hat q_{1-\alpha}.
\]
Confidence intervals are then interpolated to the SCBs as in \eqref{e:int:SCB}.
In the above formula, $\hat q_{1-\alpha}$ is the $(1-\alpha)$-quantile of the variable
\begin{align}
\label{eq:defin:hatW}
       |(\widehat \C(\cdot,\cdot))^{-1/2} \cdot\widehat W(\cdot)|_\infty, \quad \widehat W\sim\mathcal{N}(0,\widehat \C)~.
\end{align}
Quantiles as well as all empirical probabilities in our simulations are based on $1000$ simulation runs. 
In our below study, we consider the sample sizes $N=100, 200$ and fix the targeted confidence level for our SCBs at $(1-\alpha)=0.95$. Empirical coverage probabilities are reported in Table \ref{tab:confidence}.

\begin{table}[ht]
\centering
\begin{tabular}{l lcccc}
\toprule
 &  & \multicolumn{2}{c}{$N=100$} & \multicolumn{2}{c}{$N=200$} \\
\cmidrule(lr){3-4}\cmidrule(lr){5-6}
Tail 
  & $O_i\setminus X_i$
  & \eqref{e:Xind} & \eqref{e:Xdep}
  & \eqref{e:Xind} & \eqref{e:Xdep} \\
\midrule
(I)  & $\eqref{eq:defin:interval:missing}$ & 0.942 & 0.916 & 0.951 & 0.936 \\
(II) & $\eqref{eq:defin:interval:missing}$ & 0.911 & 0.911 & 0.920 & 0.921 \\
(I)  & $\eqref{eq:defin:waiting:missing}$  & 0.933 & 0.923 & 0.938 & 0.942 \\
(II) & $\eqref{eq:defin:waiting:missing}$  & 0.929 & 0.926 & 0.957 & 0.935 \\
\bottomrule
\end{tabular}
\caption{\label{tab:confidence} Empirical coverage probabilities of the SCB based on $1000$ runs and a targeted level of $(1-\alpha) =0.95$.}
\vspace{-0.5cm}
\end{table}

Results do not differ much between scenarios or between the dependent and independent case.
 For $N=200$, we observe coverage close to the nominal level in all cases. For the lower sample size $N=100$ we observe slight under-coverage. This is unsurprising, given that in the chosen setup each grid-point $t_j = j/N\in  [0,1]$ is (on average) observed roughly fifty percent of the time, leading to an effective sample size of 50  (see Figure \ref{fig:missingness2}). Such samples are of course relatively small for a Gaussian approximation.

\subsection{Multiscale Inference}

In this section, we investigate the performance of the multiscale statistics developed in Section \ref{sec:loc}. As an example, we will focus here on stationarity testing; see hypothesis $H_0^{stat}$ in \eqref{e:H0:stat} and the corresponding multiscale test statistic $\widehat{M}_N(\cdot)$ in
\eqref{e:hatM}. Threshold exceedance is considered in the next section. Stationarity testing is interesting because we can, at least in some broad sense, compare performance to an existing method, namely \cite{Hudecova:Kirch:2025} (henceforth "HK"). HK is a very difficult benchmark to reach for at least two reasons: first, HK only considers independent functional data and their method strongly exploits this assumption via permutation testing. This  boosts performance in small samples, and we point out that all sample sizes reported in HK are relatively small. \\
Second and more importantly, HK do not develop general tests against the alternative of a non-stationary mean. Rather, they develop tests that are optimized for a narrower class of non-stationarities - essentially generalized single change point models. To be precise, in their paper data are generated as follows
    \[
    X_i(t)= \mu(t)+\mu^*(t)\, g(i/N-\kappa)+ \varepsilon_i(t), \quad t\in[0,1],
    \]
where $\mu, \mu^*$ are mean functions, $g$ is a deformation over time and $\kappa$ denotes an unknown change location. Not only is the product structure of the change baked into the HK method, they even assume that the deformation $g$ is exactly known to the user. In our below setting, where we consider abrupt single changes, we have  $g(x) :=\mathbb{I}\{x > 0\}$ and 
the HK test will draw on the knowledge that $g$ is the Heaviside function. In contrast to this, no structural knowledge about the mean function is employed by our method.
We also mention that the HK method involves a choice of a tuning parameter $\gamma$, and we draw comparisons to their results for $\gamma=1/4$. This choice yields the best performance of HK in the largest number of cases. Our own test decision relies on the approximation of the quantiles under the null hypothesis, and hence, as in the previous section, on variance estimation (for a detailed description see Section \ref{var:mult} in the Supplement). 

Following the simulation setting in HK, we set $\mu\equiv 0$ and consider two scenarios for $\mu^*$
\[
    \mu^*_1(t)=0.7 \quad \textnormal{and} \quad \mu^*_2(t)=1.2e^{-2t}.
\]
The noise processes $\{\varepsilon_i\}_{i=1}^N$ are generated i.i.d. as
    \[
    \varepsilon_i(t) = \sum_{j=0}^{J} \sqrt{\lambda_j}\, \xi_{i,j} \cos(j \pi t), \quad i=1,\dots,N,
    \]
for $J=20$, $\lambda_i = 0.5 \cdot 3^{-i}$ and $\xi_{i,0},\dots,\xi_{i,J}$ i.i.d.\ standard normally distributed random variables. For the missingness mechanism we consider the three choices $(M1)-(M3)$ described in Appendix A of \cite{Hudecova:Kirch:2025}.

\begin{table}[h]
\centering
\begin{tabular}{lcccccc}
\toprule
 Method & \multicolumn{3}{c}{\eqref{e:hatM}} &\multicolumn{3}{c}{HK} \\
\toprule
Missingness  & $N=20$ & $N=50$ & $N=80$ & $N=20$ & $N=50$ & $N=80$\\
\midrule
M1 & 0.005& 0.028 & 0.059 & 0.043 &  0.030 & 0.060 \\
M2&  0.010&  0.032 & 0.061 & 0.040 & 0.043 & 0.037\\
M3 & 0.013  & 0.022  & 0.040  & 0.038 & 0.040 & 0.037\\
\bottomrule
\end{tabular}
\caption{\label{tab:confidence:M} Empirical sizes of the test \eqref{e:hatM} and of the test from \cite{Hudecova:Kirch:2025} with parameter choice $\gamma=1/4$. For both tests, the targeted significance level is $\alpha=0.05$.} 

\end{table}

\begin{table}[h]
\label{tab:spatially:uniform}
\centering
\begin{tabular}{c|c|ccc|ccc}
\toprule
 & Method & \multicolumn{3}{c}{\eqref{e:hatM}} &\multicolumn{3}{c}{HK} \\
\toprule
\textbf{$\kappa$} & Missingness& $N=20$ & $N=50$ & $N=80$ &  $N=20$ & $N=50$ & $N=80$\\
\midrule
\multirow{3}{*}{0.25} & M1 & 0.069 & 0.303 & 0.714 & 0.236 & 0.555 & 0.843 \\
                       & M2 & 0.036 & 0.209 & 0.559 & 0.225 & 0.592 & 0.849\\
                       & M3 & 0.059 & 0.299 & 0.696 & 0.211 & 0.598 & 0.870 \\
\midrule
\multirow{3}{*}{0.50} & M1 & 0.342 & 0.925 & 0.997 & 0.391 & 0.788 & 0.947\\
                       & M2 & 0.228 & 0.756 & 0.987 & 0.392 & 0.850 & 0.962 \\
                       & M3 & 0.294 & 0.878 & 0.989 & 0.419 & 0.848 & 0.960 \\
\midrule
\multirow{3}{*}{0.75} & M1 & 0.078 & 0.354 & 0.703 & 0.247 & 0.556 & 0.835\\
                       & M2 & 0.045 & 0.248 & 0.587  & 0.224 & 0.613 & 0.863\\
                       & M3 & 0.065 & 0.306 & 0.685 & 0.234 & 0.589 & 0.859 \\
\bottomrule
\end{tabular}
\caption{\label{tab:rejection:1} Rejection rates by method, missingness, and sample size.}
\end{table}

\begin{table}[h]
\label{tab:spatially:concentrated}
\centering

\begin{tabular}{c|c|ccc|ccc}
\toprule
 & Method & \multicolumn{3}{c}{\eqref{e:hatM}} &\multicolumn{3}{c}{HK} \\
\toprule
\textbf{$\kappa$} & Missingness& $N=20$ & $N=50$ & $N=80$ & $N=20$ & $N=50$ & $N=80$ \\
\midrule
\multirow{3}{*}{0.25} & M1 & 0.429 & 0.899 & 0.999 & 0.182 & 0.479 & 0.771\\
                       & M2 & 0.319 & 0.792 & 0.981 & 0.192 & 0.511 & 0.758\\
                       & M3 & 0.414 & 0.892 & 0.994 & 0.179 & 0.448 & 0.748\\
\midrule
\multirow{3}{*}{0.50} & M1 & 0.556 & 0.982 & 0.999 & 0.335 & 0.721 & 0.920\\
                       & M2 & 0.436 & 0.948 & 0.992 & 0.328 & 0.791 & 0.934\\
                       & M3 & 0.589 & 0.980 & 0.998 & 0.284 & 0.673 & 0.901\\
\midrule
\multirow{3}{*}{0.75} & M1 & 0.117 & 0.518 & 0.912& 0.185 & 0.485 & 0.750\\
                       & M2 & 0.092 & 0.391 & 0.790 & 0.186 & 0.482 & 0.776\\
                       & M3 & 0.180 & 0.575 & 0.911 & 0.166 & 0.425 & 0.703 \\
\bottomrule
\end{tabular}
\caption{\label{tab:rejection:2} Rejection rates by method, missingness, and sample size.}\vspace{-0.5cm}
\end{table}

For our evaluations, we fix the nominal level at $\alpha=0.05$ and compute (as before) all empirical probabilities using $1000$ simulation runs. Empirical rejection probabilities for HK are cited from that paper. 

Given that HK uses substantially more prior knowledge  about the data generating process and is optimized for independent samples, we are satisfied if our new, asymptotic method performs similarly. As it turns out, the new method actually often  has higher power than HK.  Rejection rates under $H_0$ are gathered in Table \ref{tab:confidence:M}, and under $H_1$ in Tables \ref{tab:rejection:1} and \ref{tab:rejection:2} respectively.  

As we can see from the tables, the nominal level is (as expected) well approximated for HK across all sample sizes. Our new method is somewhat conservative for smaller samples, giving a close approximation for $N=80$. Nevertheless, in no scenario does our new method become excessively liberal, not even for very small samples. 
Regarding the empirical power we observe that, for the alternative based on $\mu_2$,  the new method is competitive already at $N=20$ and substantially outperforms HK for the larger sample sizes for $\kappa \in \{0.25,0.5\}$. For the later change $\kappa=0.75$ and $N=50$ both methods perform similarly, but for $N=80$ our new method again dominates. The reason why our method does so well is likely because the "direction" of the change in $\mu_2^*(t)$ is largest in a small number of arguments near $t=0$. Detecting such a localized change is easy for a max-type statistic, as discussed in our paper, but difficult for an integrated statistic such as HK, which smooths out the change over $t$. 
In many applications of FDA, changes are strongly localized in $t$ making this a very important scenario. We will see a similar phenomenon in our data example, where high pollution levels are concentrated in only some hours of the day. The other alternative considered by HK is the completely uniform function $\mu_1^*\equiv 0.7$. Such perfectly uniform changes seem less realistic, but they can help to illuminate the merits of HK in this scenario. Here, HK outperforms our new method, with 
somewhat comparable performance for $N=80$, also depending on the precise change point location.

In summary, for small to moderate sample sizes ($N =20,50,  80$), our new method performs competitively with HK. We observe superior performance for the practically more relevant case of concentrated changes, while for perfectly uniform alternatives HK is superior. We have no data for HK when $N$ is larger, but our simulations strongly suggest that our new method would perform better in such a scenario.

\subsection{Threshold exceedance for pollution data} \label{sec:polution}

Monitoring air pollution is important to assess risks for vulnerable groups and as a guide to public health policy. 
Organizations such as the WHO \citep{WHO:2021} and the European Union \citep{EU:2012} have developed guidelines demarcating thresholds e.g. for the 1- and 24-hour average exposure to pollutants, enabling guideline-oriented quantitative analyses such as ours. Below, we focus on concentrations of particulate matter $PM_{10}$ (diameter of $10$ micrometers or less), which can cause both acute and chronic respiratory issues.\\
In recent years sudden spikes of $PM_{10}$ levels have been observed more frequently in some regions, with causes including wildfires, dust storms and industrial incidents  (see \cite{Deary:Griffiths:2021} for references and further exposition). Concerns over air quality have led to the publication of large, high-resolution data repositories. The statistical challenge is developing procedures that can exploit this wealth of information, while being able to handle the omnipresent problem of missing data. 
Our new methodology for testing threshold exceedance in partially observed functional time series is tailored to this very problem. 

We begin with a short description of our data. In our analysis, we focus on a time series of air quality measurements from the Bandra Kurla Complex in Mumbai, India (accessible, along with many further comparable time series at \cite{OpenAQ:2025}). 
$PM_{10}$ concentration was measured, in $\mu\text{g}/\text{m}^3$, from  19.02.2025 (dd.mm.yyyy) to  05.09.2025, in  15-minute intervals, resulting in $199$ functional observations. Sensors are offline between 23:45 and 00:30 local time, such that at most a total of $93$ observations are maximally available each day.
In practice, much fewer observations are made, because 
the sensors fail for $\approx 20\%$ of the time, with values most often missing at the end of a day. Since pollution data are continuous (also from one night to the next morning) and since missingness may overlap between days (the longest failure is about 3 days long), it is clear that both components involve some degree of temporal dependence. Therefore, tools for dependent time series, as proposed in this paper, are necessary.  More specifically, we will leverage our new methods for non-stationary data, because the time series of daily pollution profiles is certainly non-stationary over more than 6 months. Since the main interest is the detection of dangerously high pollutant levels, we apply our new test on threshold exceedance (see \eqref{e:threshold}).

Since we expect pollution levels to peak only during short periods of the observation window, we enhance our statistic's power against small-scale alternatives, following Remark \ref{rem:details} parts $II)$ and $III)$. Namely, we estimate extremal sets for better nominal level approximation and use the small scales weighting from part $III)$ with $C_0=20$. Quantiles are simulated using the long-run variance estimator in \eqref{eq_longrun_est}. The threshold $\Delta>0$ in the hypothesis \eqref{e:threshold} is user-determined and in our case demarcates a (scientifically guided) value of dangerous pollution levels. We consider the thresholds for medium, high and very high $PM_{10}$ exposure given by $\Delta=50$, $\Delta=90$ and $\Delta=180$, respectively. These values are motivated by European air--quality index categories for particulate matter $\text{PM}_{10}$. 
In commonly used European AQI/CAQI schemes, concentrations exceeding approximately $50~\mu\text{g}/\text{m}^3$ (medium), $90~\mu\text{g}/\text{m}^3$ (high), and $180~\mu\text{g}/\text{m}^3$ (very high) correspond to progressively poorer air--quality categories (\cite{EU:2012}).
These thresholds are widely used in public health communication and air--quality monitoring indices, even though they do not represent legally binding limit values under EU ambient air--quality legislation.
 We will also report the highest level $\Delta$ at which a significant threshold exceedance is detectable.  
 
Our results are as follows: The hypothesis of no threshold exceedance is rejected for $\Delta \in \{50,90,180\}$ at significance level of $0.05$. The largest $\Delta$ for which rejection took place was $651$, which is a staggering 13 times the medium threshold $\Delta=50$. To get a better sense of when air pollution is large (day of the year and time of the day), we have plotted a heatmap in Figure \ref{fig:exceedances}, in the Supplement.
  The heatmap reveals that pollution is high overnight and in the early morning hours. This pattern is often observed for particulate matter, where values are lower over the day due to vertical mixing (\cite{Girotti:2025}). We have high certainty of hazardous pollution levels in late winter and early spring. In the summer, pollution decreases somewhat and in the month of August, no significant threshold exceedances over $\Delta=90$ can be identified anymore (even though exceedances over $\Delta=50$ can be identified at almost all days and almost all times). This pattern is likely due to annual heating cycles (among other contributing factors), which lead to higher pollution in the late autumn and winter. \\[-4.5ex]

\putbib
\end{bibunit}

\newpage
\begin{bibunit}
    \appendix
\setcounter{page}{1}
\section{Supplementary Material}

In the following we define for any collection of points $\mathcal{T}\subset [0,1]$ and function $f:[0,1]\to \R$ the associated supremum norm by
\[
    |f(\cdot)|_{\infty,\mathcal{T}}=\sup_{t \in \mathcal{T}}|f(t)|.
\]
If the supremum is taken over $[0,1]$ we simply use $|\cdot|_\infty$, as before. Additionally, for two sequences $(a_N),(b_N)$ we write
\[
    a_N \lesssim b_N
\]
whenever $a_N\leq cb_n$ for some $c>0$ independent of $N$.

\subsection{Proofs: Gaussian Approximations}
For the proofs recall the definitions of $\hat \mu(\cdot)$ and $\hat N(\cdot)$ (both in eq. \eqref{e:def:pi}) and $\widehat{T}_N$ (see eq. \eqref{e:hat:T}). Additionally, we define
\begin{align}
    S_N(t):=&\frac{1}{N(t)}\sum_{i=1}^N \{X_i(t)O_i(t)\}-\mu(t), \quad N(t):=N\pi(t),\\
    \hat S_N(t):=&\frac{1}{\hat N(t)}\sum_{i=1}^N \{X_i(t)O_i(t)\}-\mu(t),\\
    \sigma_{\textbf{O}}(t,s):=&\sum_{k=-\infty}^\infty \C[ O_0(t),O_k(s)],\quad \sigma_{\textbf{O}}^2(t):=\sigma_{\textbf{O}}(t,t)
\end{align}
and, for a set of points  $\{t_1,...,t_p\}$, also the discretized maximum
\[
    T_N=\sqrt{N}\max_{1 \leq i \leq p}|S_N(t_i)|~.
\]
Notice that $T_N$ is essentially the same as $\widehat{T}_N$, with $S_N$ replaced by $\hat S_N$ in the latter. 

During the proofs, we will frequently say that we "apply a high-dimensional Gaussian approximation" (such as Corollary 2.2 from \cite{zhang:cheng:2018}) to a vector $(S_N(t_1),...,S_N(t_p))$ to obtain a result about the distribution of $\max_{1 \leq i \leq p}|S_N(t_i)|$. In this case we mean applying the result to the vector $(S_N(t_1),...,S_N(t_p), -S_N(t_1),...,-S_N(t_p))$ which always fulfills the approximation's assumptions, whenever the original vector $(S_N(t_1),...,S_N(t_p))$ does. The desired approximation then follows by noting that
\[
    \max_{1 \leq i \leq p}\max\{S_N(t_i),-S_N(t_i)\}=\max_{1 \leq i \leq p}|S_N(t_i)|~.
\]

\begin{lem}
    \label{lem:gauss:approx:missingness}
    Let $p=p(N)$ be a sequence satisfying $p \sim N^{C^*}$ for some $C^*>0$ and let $t_i=t_{i}(N) \in [0,1], 1 \le i \le p$ be an array of points. 
    Let
    \[
        T_N^O:=\max_{j=1,...,p}N^{-1/2}\sum_{i=1}^N\Big(O_i(t_j)-\pi(t_j)\Big).
    \]    
    Then there exists a sequence of mean-zero Gaussian vectors $W_N\in \R^p$ such that for any $x_0>0$ we have
    \begin{align}
        \sup_{t \geq x_0}|\PR(T_N^O\geq t)-\PR(\max_{1 \leq j \leq p}W_{N,j}\geq t)|&=o(1)~\\
         \C[W_{N,i},W_{N,j}]&=\sigma_{\textbf{O}}(t_i,t_j)~.
    \end{align}
\end{lem}
\begin{proof}
    The proof follows by similar but easier arguments as those in the proof of Theorem \ref{theo:gauss:approx:multiscale} (coordinates with small variances are negligible, the remaining coordinates may be handled via Theorem \eqref{thm:zhang:cheng:extension}) and we therefore omit the details.

\end{proof}

\begin{lem}
\label{lem:gauss:approx:confidence}
    Let $p=p(N)$ be a sequence satisfying $p \sim N^{C^*}$ and let $t_i=t_{i}(N) \in [0,1], 1 \le i \le p$ be an array of points.
    Then there exists a sequence of mean-zero Gaussian vectors $W_N \in \R^p$ such that 
    \begin{align}
        \sup_{t \in \R}|\PR(\widehat{T}_N\geq t)-\PR(\max_{1 \leq i \leq p}|W_{N,i}|\geq t)|&=o(1)~\\
         \C[W_{N,i},W_{N,j}]&=(\pi(t_i)\pi(t_j))^{-1}\sigma_{\textbf{Y}}(t_i,t_j)~.
    \end{align}
\end{lem}
\begin{proof} We first show the desired result for $T_N$ and then transfer it to $\widehat{T}_N$ by an approximation argument. We want to apply Corollary 2.2 from \cite{zhang:cheng:2018} for which we need to check its assumptions. Adopting their notation we hence need to show that there exist constants $\mathfrak{D}_N,c_1,c_2,c_3,K_1,\delta>0$ such that 
\begin{align}
\label{ass:1}
    \max_{1 \leq i \leq N}\E[\max_{1 \leq j \leq p}Y_i(t_j)^4/\mathfrak{D}_N^4]&\leq K_1, \quad \mathfrak{D}_N \leq N^{3/32-\delta}\\ 
    \label{ass:2}
    c_1 \leq \min_{1 \leq i \leq p}\C[S_N(t_i),S_N(t_i)]&\leq \max_{1\leq i \leq p}\C[S_N(t_i),S_N(t_i)]\leq c_2\\ \label{ass:3}
    \sum_{j=1}^\infty \max_{1 \leq i \leq p}j\theta_{j,i,3}&<c_3
\end{align}
where $\theta_{j,i,3}$ is defined in equation (6) in \cite{zhang:cheng:2018}. This particular form of the assumptions is obtained from their more general result, by choosing $h(x)=x^4$ and using that $\log(p)\lesssim \log(N)$ in our case. We validate the assumptions in turn.

\begin{itemize}
    \item[Eq. ]\eqref{ass:1} : Choose $\mathfrak{D}_N=1$ and note that 
        \[
        \max_{1 \leq i \leq N} \E[\max_{1 \leq j \leq p}(X_i(t_j)O_i(t_j))^4]\leq \E[|X_1(\cdot)|_\infty^4]<\infty~,
        \]
        by Assumption (A).
    \item[Eq. ] \eqref{ass:2}: Using Lemma \ref{lem:physic:dep} (iv) it is easy to check that $\C[S_N(t_i),S_N(t_i)]$ converges (uniformly in $i$) to $\sigma_{\textbf{Y}}^2(t_i)$ so that it suffices to lower/upper bound this quantity by $c_1$ and $c_2$, respectively. The lower bound is an immediate consequence of Lemma \ref{lem:lrv:formulas} (ii) in combination with Assumption (B).  We may apply this lemma because $\min_{0 \leq t \leq 1}\pi(t)>0$ and Assumption (E) ensure that its requirements are met. The upper bound follows from using Lemma \ref{lem:physic:dep} (iv) to upper bound the expressions for $\sigma_{\bf Y}^2=\sigma^2_{\bf XO}$ in  Lemma \ref{lem:lrv:formulas} (i).
    \item[Eq. ] \eqref{ass:3}:  Follows immediately from Assumption (C).
\end{itemize}

We may therefore apply the above-mentioned Corollary 2.2 in \cite{zhang:cheng:2018} which implies that there exists a sequence of $p$-dimensional Gaussian vectors $(Z_{N,1},...,Z_{N,p})$ with 
    \[
        \C[Z_{N,i},Z_{N,j}]=(N\pi(t_i)\pi(t_j))^{-1}\C\Big[\sum_{l=1}^NY_l(t_i),\sum_{l=1}^NY_l(t_j)\Big]
    \]
    that satisfies
    \begin{align}
    \label{eq:TN:gauss:approx}
        \sup_{t \in \R}|\PR(T_N\geq t)-\PR(\max_{1 \leq i \leq p}|Z_{N,i}|\geq t)|=o(1).
    \end{align}       
    Using that 
    \[
        \C[Y_i(t_j),Y_k(t_m)]\lesssim \chi^{|i-k|}
    \]
    as well as Lemma \ref{lem:lrc:partial:sums} we conclude that $\C[Z_{N,i},Z_{N,j}]$ converges fast enough to apply Theorem 2 from \cite{chernozhukov:chetverikov:kato:2015} to find another Gaussian vector $(W_{N,1},...,W_{N,p})$ with
    \[
      \C[W_{N,i},W_{N,j}]=(\pi(t_i)\pi(t_j))^{-1}\sigma_{\textbf{Y}}(t_i,t_j),
    \]
    such that
    \[
        \sup_{t \in \R}|\PR(T_N\geq t)-\PR(\max_{1 \leq i \leq p}|W_{N,i}|\geq t)|=o(1)~.
    \]
    Next, we note that 
    \begin{align}
    \label{eq:empirical:count:confidence}
         \PR(\widehat{T}_N\geq t)&=\PR(T_N\geq t-(\widehat{T}_N- T_N))\\
         &\leq\PR(T_N\geq t-N^{-1/2+\delta})+o(1)\\
         &=\PR(\max_{1 \leq i \leq p}|W_{N,i}|\geq t-N^{-1/2+\delta})+o(1)\\
         &=\PR(\max_{1 \leq i \leq p}|W_{N,i}|\geq t)+o(1)
    \end{align}
       where the second line follows by Lemma \ref{lem:approx:Nx}, the third by the previous arguments, and the last by Theorem 3 from \cite{chernozhukov:chetverikov:kato:2015}. An analogous argument yields the reverse inequality and thus also the desired result.
\end{proof}

\begin{lem} 
\label{lem:gauss:tight}
    The Gaussian vectors $W_N$ in Lemmas \ref{lem:gauss:approx:missingness} and \ref{lem:gauss:approx:confidence} can be realized as finite-dimensional evaluations of continuous Gaussian processes $W_O,W_Y \in \mathcal{C}[0,1]$. In particular, irrespective of the choice of $t_1,...,t_p$, we have that, for $W_N$ in Lemma \ref{lem:gauss:approx:missingness}, that
    \[
        |W_N|\leq |W_O|_\infty=O_{\PR}(1),
    \]
    and, for $W_N$ in Lemma \ref{lem:gauss:approx:confidence}, that
    \[
       |W_N|\leq  |W_Y|_\infty=O_{\PR}(1)~.
    \]
\end{lem}
\begin{proof}
    We only handle the more difficult case for the vectors in Lemma \ref{lem:gauss:approx:confidence}; the other follows analogously. Existence of a Gaussian process defined on $[0,1]$ with covariances
    \[
        \C[W_Y(t),W_Y(s)]=(\pi(t)\pi(s))^{-1}\sigma_{\textbf{Y}}(t,s)
    \]
    follows easily from the Kolmogorov extension theorem and we thus need only verify its continuity. Note that for any integer $q$ it holds that
    \[
        \E[|W_Y(t)-W_Y(s)|^q]\leq C_q\E[|W_Y(t)-W_Y(s)|^2]^{q/2}
    \]
    so that we obtain the sample path continuity by the Kolmogorov-Chentsov theorem if $\E[|W_Y(t)-W_Y(s)|^2]\leq C|t-s|^\eta$ for some $C,\eta>0$. This in turn follows easily if we can show that $\sigma_{\mathbf{Y}}(t,s)$ is $\eta$-Hölder continuous. As an example we show that
    \[
        |\sigma_{\mathbf{Y}}(t,t)-\sigma_{\mathbf{Y}}(s,s)|\lesssim |t-s|^\eta
    \]
    and note that the general case follows by similar but notationally more laborious arguments. Using Lemma \ref{lem:lrv:formulas} and the fact that $Y_i(t)=X_i(t)O_i(t)$ we need to show that the functions $\E[X_i(\cdot)], \E[O_i(\cdot)],\sigma_{X(\cdot)}^2,\sigma_{O(\cdot)}^2$ and
    \[
        \sum_{i \in \Z}\C[X_0(\cdot)X_i(\cdot)]\C[O_0(\cdot)O_i(\cdot)]
    \]
   are $\eta$-Hölder continuous. Continuity of the expectations follow by Assumption (B). The other three terms can be handled similarly and we only inspect the last, most complicated one.  To that end we first notice that, by  Lemma \ref{lem:physic:dep} (v) and Assumption (D)
   \begin{align*}
       |\C[X_k(t),X_0(t)]-\C[X_k(s),X_0(s)]|&\leq |\C[X_k(t)-X_k(s),X_0(t)]|+|\C[X_k(s),X_0(s)-X_0(t)]|\\
       &\lesssim |t-s|^{\vartheta}\chi^k
   \end{align*}
   and similarly also
   \[
    |\C[O_0(t),O_k(t)]-\C[O_0(s),O_k(s)]|\lesssim |t-s|^\vartheta\chi^k~,
   \]
   from which the desired continuity follows by a classical triangle inequality argument and the summability of $\chi^k$.
\end{proof}

\begin{lem}
\label{lem:approx:Nx}
    Let $\mathcal{T}=\mathcal{T}_N$ be any sequence of finite subsets of $[0,1]$ such that $|\mathcal{T}|\lesssim N^{C^*}$ for some $C^*>0$. Then, for any $\delta>0$, it holds that
    \begin{align}
        \sqrt{N}|\hat S_N(\cdot)-S_N(\cdot)|_{\infty,\mathcal{T}}=O_{\PR}(N^{-1/2+\delta})
    \end{align}
\end{lem}
\begin{proof}
    We begin by observing that
    \begin{align*}
        \sqrt{N}|\hat S_N(\cdot)-S_N(\cdot)|_{\infty,\mathcal{T}} \leq& \Bigg|\frac{N(\cdot)-\hat N(\cdot)}{\hat N(\cdot)}\Bigg|_{\infty,\mathcal{T}}\Bigg|N^{-1/2}\sum_{i=1}^{N}X_i(\cdot)O_i(\cdot)\Bigg|_{\infty,\mathcal{T}}
    \end{align*}
    Using that $\inf_{t \in [0,1]}\pi(t)>0$ as well as Lemma \ref{lem:gauss:approx:confidence} and eq. \eqref{eq:TN:gauss:approx}, we obtain that the second term on the right hand side is $O_{\PR}(1)$. The first factor can be handled similarly, using Lemma \ref{lem:gauss:approx:missingness} in place of Lemma \ref{lem:gauss:approx:confidence} to obtain that, with high probability
    \[
        \hat N(\cdot)/N\geq \inf_{t \in [0,1]}\pi(t)/2
    \]
    as well as 
    \[
    N^{-1}\Big|N(\cdot)-\hat N(\cdot)\Big|_{\infty,\mathcal{T}}=O_{\PR}(N^{-1/2+\delta})~,
    \]
    which yields the desired result.
\end{proof}

\subsection{Proof: Theorem \ref{theo:confidence:bands}}

\begin{proof}
Without loss of generality we consider the case where $0=t_1$ and $ t_p=1$. Let us begin by noting that $\mu(\cdot) \notin [\hat C^-(\cdot),\hat C^+(\cdot)]$ is equivalent to there existing some $s \in [0,1]$ and some  $i \in \{1,...,p-1\}$ such that either 
\begin{align}
\label{eq:mu:lower:deviation}
    \mu(st_i+(1-s)t_{i+1})<s\hat C^-(t_i)+(1-s)\hat C^-(t_{i+1}),
\end{align}
or
\begin{align}
\label{eq:mu:upper:deviation}
     \mu(st_i+(1-s)t_{i+1})>s\hat C^+(t_i)+(1-s)\hat C^+(t_{i+1}).
\end{align}
For illustration, let us consider the case eq. \eqref{eq:mu:lower:deviation}. By Assumption (B) we have that $\mu$ is $\vartheta$-Hölder continuous, implying that there exists a constant $K_1$ such that
\[
    |\mu(st_i+(1-s)t_{i+1})-\mu(t_i)|\leq K_1 |t_i-t_{i+1}|^\vartheta\leq C_1N^{-C^*\vartheta}.
\]
Therefore \eqref{eq:mu:lower:deviation} implies
\[
    \mu(t_i)<s(\hat C^-(t_i)+K_1N^{-C^*\vartheta})+(1-s)(\hat C^-(t_{i+1})+K_1N^{-C^*\vartheta})~.
\]
For this to hold true it is necessary that either $\mu(t_i)<\hat C^-(t_i)+K_1N^{-C^*\vartheta}$ or $\mu(t_i)<\hat C^-(t_{i+1})+K_1N^{-C^*\vartheta}$. Again using the continuity of the $\mu(\cdot)$ we thus have that \eqref{eq:mu:lower:deviation} implies that one of the inequalities 
\begin{align*}
    \mu(t_i)&<\hat C^-(t_i)+2K_1N^{-C^*\vartheta}\\
    \mu(t_{i+1})&<\hat C^-(t_{i+1})+2K_1N^{-C^*\vartheta}
\end{align*}
holds. A similar argument as for the case \eqref{eq:mu:upper:deviation} then yields that $\Big\{\mu(\cdot) \notin [\hat C^-(\cdot),\hat C^+(\cdot)]\Big\}$ is a subset of
\[
 A_N=\Big\{\exists i \in \{1,...,p\}: \mu(t_i)\notin [\hat C^-(t_i)+2K_1N^{-C^*\vartheta},\hat  C^+(t_i)-2K_1N^{-C^*\vartheta}] \Big\}~.
\]
Rewriting this set in terms of $\hat S_N$ and using Lemma \ref{lem:gauss:approx:confidence} we thus obtain 
\begin{align*}
    \PR(A_N)&=\PR(\max_{1 \leq i \leq p}|\hat S_N(t_i)|>q_{1-\alpha} N^{-1/2}+2K_1N^{-C^*\vartheta})\\
    &=\PR(\max_{1 \leq i \leq p}|W_{N,i}|>q_{1-\alpha}+2K_1N^{-C^*\vartheta+1/2})+o(1).
\end{align*}
Next we apply Theorem 3 from \cite{chernozhukov:chetverikov:kato:2015} (Lemma \ref{lem:lrv:formulas} ensures the required bounds for the variances of the Gaussian $W_N$, $C^*\vartheta>1/2$ that the resulting bound goes to 0) to obtain that 
\[
    \PR(\max_{1 \leq i \leq p}|W_{N,i}|>q_{1-\alpha}+2K_1N^{-C^*\vartheta+1/2})=\PR(\max_{1 \leq i \leq p}|W_{N,i}|>q_{1-\alpha})+o(1),
\]
from which the desired result follows by choosing $q$ as the $(1-\alpha)$-quantile of $|W_Y|_\infty$ in Lemma \ref{lem:gauss:tight}.
\end{proof}

\subsection{Lemmas about Covariances and Moments}

Here and in the following we will, for a stationary sequence $\textbf{X}=(X_i)$, define its expected value by $\E[\textbf{X}]:=\E[X_0]$ and its long-run variance by $\sigma^2_{\textbf{X}}:=\sum_{i\in \Z}\E[(X_i-\E[\mathbf{X}])(X_0-\E[\mathbf{X}])]$~.

\begin{lem}
\label{lem:lrv:formulas}
    Consider four jointly stationary sequences $\textbf{X}=(X_i),\textbf{O}=(O_i),\textbf{W}=(W_i)$ and $\textbf{Z}=(Z_i)$ of random variables for which $\textbf{X}$ is independent of $\textbf{O}$ and $\textbf{W}$ is independent of $\textbf{Z}$. Further assume that their autocovariances  and those of their products are absolutely summable. Then
    \begin{itemize}
        \item[(i)] \begin{align*}        \sigma_{\bf XO,WZ}&=\E[\bf X]\E[W]\sigma_{\bf O,Z}+\E[O]\E[Z]\sigma_{\bf X,W}\\& \quad \quad+\sum_{i \in \Z}\C[X_0,W_i]\C[O_0,Z_i]+\sum_{i \in \Z}\C[X_0W_{i},O_0Z_{i}]~.        
    \end{align*}   
    In particular, if $\textbf{X}=\textbf{W}$ and $\textbf{O}=\textbf{Z}$, we have 
    \begin{align*}
    \sigma_{\bf XO}^2=\E[\textbf{O}]^2\sigma_{\textbf{X}}^2+\E[\textbf{X}]^2\sigma_{\bf O}^2+{\sum_{i \in \Z}\C[X_0,X_i]\C[O_0,O_i]}
    \end{align*}
    \item[(ii)] Assuming that  $\bf\E[O]^2\sigma_X^2>0$, we have
    \[
      \bf\sigma_{XO}^2>0~.
    \]
    \end{itemize}
\end{lem}
\begin{proof}

    Group the product as
    \[
    X_tO_tW_{t-h}Z_{t-h}
    =(X_tW_{t-h})(O_tZ_{t-h}),
    \]
    so
    \[
    \mathbb{E}[X_tO_tW_{t-h}Z_{t-h}]
    =
    \C[X_tW_{t-h},\,O_tZ_{t-h}]
    +\mathbb{E}[X_tW_{t-h}]\,\mathbb{E}[O_tZ_{t-h}] .
    \]
    Next observe that
    \[
    \mathbb{E}[X_tW_{t-h}]=\E[\textbf{X}]\E[\textbf{W}]+\C[X_t,W_{t-h}],
    \qquad
    \mathbb{E}[O_tZ_{t-h}]=\E[\textbf{O}]\E[\textbf{Z}]+\C[O_t,Z_{t-h}]~.
    \]
    Therefore
    \begin{align*}
        \mathbb{E}[X_tO_tW_{t-h}Z_{t-h}]
    &=
    \C[X_tW_{t-h},O_tZ_{t-h}]
    +\E[\textbf{X}]\E[\textbf{O}]\E[\textbf{W}]\E[\textbf{Z}]
    +\E[\textbf{X}]\E[\textbf{W}]\C[O_t,Z_{t-h}]\\
    &\quad\quad +\E[\textbf{O}]\E[\textbf{Z}]\C[X_t,W_{t-h}]
    +\C[X_t,W_{t-h}]\C[O_t,Z_{t-h}].
    \end{align*}    
    Now we subtract $\E[\textbf{X}]\E[\textbf{O}]\E[\textbf{W}]\E[\textbf{Z}]$ on both sides and sum over all lags to obtain part (i) of the lemma.\\  
   For the second part we observe that for $\tilde{\mathbf{X}} = \mathbf{X}-\E\mathbf{X}$ 
   and $\tilde{\mathbf{O}}=\mathbf{O}-\E\mathbf{O}$ we have
   $0 \le \sigma_{\bf \tilde X \tilde O}^2 = \sum_{i \in \Z}\C[\tilde X_0,\tilde X_i]\C[\tilde O_0,\tilde O_i] = \sum_{i \in \Z}\C[X_0,X_i]\C[O_0,O_i]$. The desired conclusion then follows from the second part of (i).
\end{proof}

For a sequence  $\{\mathcal{F}_i\}_{i \in \Z}=\{(...,\eta_{i-1},\eta_i)\}_{i \in \Z}$ consisting of iid random variables, we denote by 
\[
    P_iX=\E[X|\mathcal{F}_i]-\E[X|\mathcal{F}_{i-1}]
\]
whenever $X$ is integrable. Note that here $X$ need not be generated as in equation \eqref{eq:filter:formalism}, further we also define
\[
    \mathcal{F}_i'=(...,\eta_{-1},\eta_0',\eta_1,...,\eta_i)
\]
where $\eta'_0$ is an independent copy of $\eta_0$. For the convenience of the reader we also record some properties of the projections $P_i$ and the resulting consequences that are well known in the literature:
\begin{lem}
\label{lem:physic:dep}
    \begin{enumerate}
        \item[(i)] $P_iX$ is uncorrelated with any $\mathcal{F}_{i-1}$-measurable square integrable random variable.
        \item[(ii)] If $X$ is $\mathcal{F}_k$ measurable we have that 
        \[
            X-\E[X]=\sum_{i=0}^\infty P_{k-i}(X)
        \]
        \item[(iii)] For some  filter $G(\mathcal{F}_i):\mathcal{S}^\Z\to \R$ let
        \[
           \delta_q(G,i)=\|G(\mathcal{F}_i)-G(\mathcal{F}_i')\|_q~.    
        \]
        From now on assume that $X_i$ is generated as follows $X_i=G(\mathcal{F}_i)$. Then, we have
        \[
            \|P_0(X_i)\|_q \leq \delta_q(G,i)
        \]
        \item[(iv)] Assume that for some monotonically decreasing function $f$ with $(f(n))_{n \in \N}\in \ell_1$ we have
        \[
            \delta_q(G,i)\leq f(i).
        \]
        Then
        \[
            |\C[X_0,X_{k}]|\leq Kf(k)
        \]
        for some constant $K$ that depends only on $f$. 
        \item[(v)] For some filter $G(t,\mathcal{F}_i):[0,1] \times \mathcal{S}^\Z \to \R$ and any $0<\alpha\leq 1$ let
        \[
            \tau_2(G,i):=\sup_{s,t \in [0,1]}|t-s|^{-\alpha}\|G(t,\mathcal{F}_i)-G(s,\mathcal{F}_i)-(G(t,\mathcal{F}_i')-G(s,\mathcal{F}_i'))\|_2~.
        \]
       Assume that for some monotonically decreasing function $f$ with $(f(n))_{n \in \N}\in \ell_2$ we have
        \[
            \tau_2(G,i)\leq f(i)~.
        \]
        Then 
        \[
           |\C[G(t,\mathcal{F}_i)-G(s,\mathcal{F}_i),G(t,\mathcal{F}_0)-G(s,\mathcal{F}_0)]| \leq K |t-s|^\alpha f(i)
        \]
        for some constant $K$ that depends only on $f$.
    \end{enumerate}
\end{lem}
\begin{proof}
    \begin{enumerate}
        \item[(i)] Follows immediately by noting that $X\to \E[X|\mathcal{F}_{i-1}]$ is the orthogonal projection onto the space of $\mathcal{F}_{i-1}$-measurable square integrable random variables.
        \item[(ii)] Note that 
        \[
            \sum_{i=0}^m P_{k-i}(X)=X-\E[X|\mathcal{F}_{k-m-1}]=:X-Q_{m}X.
        \]
        Clearly $(Q_mX,\mathcal{F}_{k-m-1})$ is a reverse martingale and thus converges almost surely to some random variable $G$. By construction $G$ is measurable with respect to the terminal sigma algebra of $(\eta_i)_{i \in \Z}$ and thus constant by Kolmogorov's 0-1 law. By taking expectations of the above equation we observe that the limit must be $G=\E[X]$, finishing the proof. 
        \item[(iii)] We begin by noting that 
        \[
            \E[X_i|\mathcal{F}_{-1}]=\E[X_i'|\mathcal{F}_{-1}]=\E[X_i'|\mathcal{F}_0]
        \]
        so that 
        \[
            P_0(X_i)=\E[X_i-X_i'|\mathcal{F}_0]~.
        \]
        The conclusion follows by the fact that conditional expectations are $L_q$-contractions.
        \item[(iv)] 
        WLOG we assume that the means are zero. Observe that 
        \begin{align*}
           \E[X_0X_{k}]&=\E[\sum_{i=0}^\infty X_0P_{k-i}X_k]
            =\sum_{i \ge k}^\infty\E[X_0P_{k-i}X_k]\\
           &=  \sum_{i \ge 0}^\infty\E[X_0P_{-i}X_k] =  \sum_{i \ge 0}^\infty\E[P_{-i}X_0P_{-i}X_k] \le \sum_{i \ge 0}^\infty\|P_{-i}X_0\|_2 \|P_{-i}X_k\|_2\\
            &\leq \sum_{i=0}^\infty\delta_2(k+i)\delta_2(i)\\
            &\lesssim f(k)
        \end{align*}
        The first line follows by (ii), the second by (i) and the fourth by (iii). 
        \item[(v)] We proceed as in the proof of (iv) with the only difference being that  in the third line we use the bound
        \begin{align*}
           & \|P_{-i}(G(t,\mathcal{F}_0)-G(s,\mathcal{F}_0))\|_2\leq \|G(t,\mathcal{F}_i)-G(s,\mathcal{F}_i)-(G(t,\mathcal{F}_i')-G(s,\mathcal{F}_i'))\|_2\\
            \leq & \tau_2(G,i)|t-s|^\alpha  
        \end{align*}
        instead of (iii). The first inequality follows by the same arguments as in the proof of (iii), while the second follows by the definition of $\tau_2(G,i)$.
    \end{enumerate}
    
\end{proof}

\begin{lem}
\label{lem:lrc:partial:sums}
    Let $\textbf{X}$ and $\textbf{O}$ be two centered, jointly stationary sequences and define 
    \[
        S_l^k=\sum_{i=l+1}^kX_i, \quad T_l^k=\sum_{i=l+1}^kO_i
    \]
    and assume that for some $C>0$ 
    \[
        |\E[X_iO_k]|=:|\gamma_{\bf O,X}(k-i)|\leq C\chi^{|k-i|}~.
    \]
    Then, for $l<m$ and letting $N_1=\min\{k-l,h-m\}$, we have
    \[
        |N_1^{-1}\C[S_l^k,T_m^h]-(1-(m-l)/N_1)_+\sigma_{\bf X,O}|=O(N_1^{-1})~.
    \]
    where the hidden constants depend only on $C$. 
\end{lem}
\begin{proof}
    Assume WLOG that $N_1\leq N_2$ and $N_1=k-l,N_2=h-m$. We exemplarily inspect the case $0\leq d=m-l\leq N_1, N_2\geq 2N_1$ as it is the most involved one. All others follow by similar arguments. We begin by expanding 
    \begin{align*}
        \C[S_l^k,T_m^h]&=\sum_{\substack{l+1 \leq  i \leq k\\ m+1 \leq j \leq h  }}\C[X_i,O_j]\\
        &=\sum_{s=0}^{N_1-1}\sum_{r=0}^{N_2-1}\gamma_{\bf O,X}(d+r-s)\\
        &=\sum_{\tau=-N_1+1}^{N_2-1}c(\tau)\gamma_{\bf O,X}(d+\tau)\\
        &=\sum_{\tau=d-N_1+1}^{d+N_2-1}c(\tau-d)\gamma_{\bf O,X}(\tau)\\
        &=\sum_{\tau=d-N_1+1}^{d+N_1-1}c(\tau-d)\gamma_{\bf O,X}(\tau)+\sum_{\tau=d+N_1}^{d+N_2-1}c(\tau-d)\gamma_{\bf O,X}(\tau)\\
        &=:I_1+I_2
    \end{align*}
    where
    \[
            c(\tau)=\begin{cases}
            N_1 \land (N_2-\tau), &\quad \tau \ge 0\\
            N_1+\tau, &\quad \tau<0~.
        \end{cases}
    \]
    Now let us analyze each term separately. First we show that $I_2$ is negligible:
    \begin{align*}
        &\sum_{\tau=d+N_1}^{d+N_2-1}c(\tau-d)|\gamma_{\bf O,X}(\tau)|\\
         = &\sum_{\tau=N_1}^{N_2-1}c(\tau)|\gamma_{\bf O,X}(\tau+d)|\\
        \le&\sum_{\tau=N_1}^{N_2-1}(N_2-\tau)|\gamma_{\bf O,X}(\tau+d)|\\
        =&\sum_{\tau=0}^{N_2-N_1-1}(N_2-N_1-\tau)|\gamma_{\bf O,X}(\tau+d+N_1)|\\
        \leq & C\chi^{N_1}
    \end{align*}
    for some $C>0$. Dividing by $N_1$ yields that $I_2=o(N_1^{-1})$\\

    For $I_1$ we analyze the cases $\tau<d$ and $\tau\geq d$ separately. For the former we have $c(\tau)=N_1-|\tau-d|$, yielding
    \begin{align*}
        &\Big|\sum_{\tau=d-N_1+1}^{d-1}(N_1-|\tau-d|)\gamma_{\bf O,X}(\tau)-\sum_{\tau=d-N_1+1}^{d-1}(N_1-d)\gamma_{\bf O,X}(\tau)\Big|\\
        \leq&\sum_{\tau=d-N_1+1}^{d-1}|\tau \cdot \gamma_{\bf O,X}(\tau)|~.
    \end{align*}
    Similarly, noting that $c(\tau)=N_1$ for $\tau\geq 0$ because $N_2\geq 2N_1$, we have
    \begin{align*}
        &\Big|\sum_{\tau=d}^{d+N_1-1}N_1\gamma_{\bf O,X}(\tau)-\sum_{\tau=d}^{d+N_1-1}(N_1-d)\gamma_{\bf O,X}(\tau)\Big|\\
        \leq&\sum_{\tau=d}^{d+N_1-1}d|\gamma_{\bf O,X}(\tau)|\leq Cd\chi^d.
    \end{align*}    
    Dividing by $N_1$ and summing up the inequalities yields the desired result if we can show that
    \[
        N_1^{-1}\sum_{\tau=d-N_1+1}^{d+N_1-1}(N_1-d)\gamma_{\bf O,X}(\tau)=(1-d/N_1)_+\sigma_{\bf X,O}+O(N_1^{-1})~.
    \]
    To that end observe
    \begin{align*}
        &\Big|\sum_{\tau=d-N_1+1}^{d+N_1-1}(N_1-d)\gamma_{\bf O,X}(\tau)-(N_1-d)\sigma_{\bf X,O}\Big|\\
        \leq&N_1 \sum_{\tau=d+N_1}^\infty|\gamma_{\bf O,X}(\tau)|+(N_1-d)\sum^{d-N_1}_{\tau=-\infty}|\gamma_{\bf O,X}(\tau)|\\
        \lesssim & N_1\chi^{N_1}+(N_1-d)\chi^{N_1-d}=O(1)
    \end{align*}
    which yields the desired bound upon dividing by $N_1$.
\end{proof}

\subsection{Dependent Missingness}
\begin{lem}
    The missingness indicator in \eqref{eq:defin:waiting:missing} satisfies Assumptions (B) - (D). 
\end{lem}
\begin{proof}
    \begin{itemize}
        \item[(B)] We will show that $\pi(t)=\PR(O_i(t)=1)$ is Lipschitz continuous. To that end we note that
        \[
            \pi(t)=\PR(O_i(t)=1,O_{i-1}(1)=0)+\PR(O_i(t)=1,O_{i-1}(1)=1)
        \]
       Observe that $O_i(t)=1, O_{i-1}(1)=0$ holds iff $T_{i,2j}\leq t \leq T_{i,2j+1}$ for some $j$ and $O_{i-1}(1)=0$ holds. An analogous statement holds for the event in the rightmost probability. Using the independence of the $\eta_i$ we hence merely need to show that the sums
        \[
            \sum_{k=1}^\infty\PR(S_{i,2k}\leq t \leq S_{i,2k+1}), \quad \sum_{k=1}^\infty\PR(T_{i,2k}\leq t \leq T_{i,2k+1}) 
        \]
        are Lipschitz in $t$. Let $p_{i,k}$ and $q_{i,k}$ denote the densities of $S_{i,k}$ and $T_{i,k}$ and $\lambda_{\min}$ and $\lambda_{\max}$ a lower and upper bound for the exponential intensities, respectively. Note that 
        \begin{align*}
            \PR(S_{i,2k}\leq t \leq S_{i,2k+1})&=\int_0^tp_{i,2k}(x)\PR(E_{i,2k+1}>t-x|S_{i,2k}=x)dx\\
            &=\int_0^tp_{i,2k}(x)e^{-\lambda_{2k+1}(t-x)}dx
        \end{align*}

        and that $\lambda_{\max}<\infty$. It therefore suffices to show that (recall $t \in [0,1]$)
        \[
            \sum_{k=1}^\infty |p_{i,k}|_{\infty, [0,1]}+|q_{i,k}|_{\infty, [0,1]}<\infty~.
        \]
        This follows by noting that
        \[
            p_{i,k}(t)\leq \frac{(\lambda_{\max}t)^{k-1}}{(k-1)!}e^{-\lambda_{\min}t}
        \]
        which follows by a straightforward induction using that the density of $S_{i,k}$ is the convolution of exponential densities. The statement for $q_{i,k}$ follows analogously. 
        \item[(C)] Define
        \[
            \Delta_i(t)=G_2(t,\mathcal{F}_i)-G_2(t,\mathcal{F}_i')=:O_i(t)-O'_i(t)
        \]
        and note that the desired conclusion follows if we can show that
        \[
            \PR(\Delta_i(t)\neq 0)\leq \chi^i
        \]
        for some $\chi<1$. By definition of $O_i$ the event $\{\Delta_i(t) \neq 0\}$ is a subset of the event $\{O_{i-1}(1)\neq O_{i-1}'(1)\}$. Hence
        \begin{align*}
            \PR(\Delta_i(t)\neq 0)&\leq \PR(O_{i-1}(1)\neq O_{i-1}'(1))\\
            &=1-\PR(O_{i-1}(1)= O_{i-1}'(1))\\
            &=1-\PR(O_{i-2}(1)=O_{i-2}'(1))-\PR(O_{i-2}(1)\neq O_{i-2}'(1),O_{i-1}(1)= O_{i-1}'(1))~.
        \end{align*}
        Next we observe that $\{O_{i-2}(1)\neq O_{i-2}'(1),O_{i-1}(1)= O_{i-1}'(1)\}$ may be rewritten as $A_i\cap\{O_{i-2}(1)\neq O_{i-2}'(1)\}$. Here
        \begin{align*}
            A_i=\{\exists j=2k+1, k\in \N : \text{Swap}(O_{i-1})=\text{Swap}(O_{i-1})+j\}
        \end{align*}
        where $\text{Swap}(O_i)$ counts the number of jumps of the function $O_i(\cdot)$. This set is $\sigma(\eta_{i-1})$-measurable. Hence, by independence,
        \begin{align*}
            &1-\PR(O_{i-2}(1)=O_{i-2}'(1))-\PR(O_{i-2}(1)\neq O_{i-2}'(1),O_{i-1}(1)=O_{i-1}'(1))\\
            =&1-\PR(O_{i-2}(1)=O_{i-2}'(1))-\PR(A_i)\PR(O_{i-2}(1)\neq O_{i-2}'(1))\\
            =&(1-\PR(A_i))\PR(O_{i-2}(1)\neq O_{i-2}'(1))
        \end{align*}
        Iterating this argument the desired result follows with $\chi=(1-\PR(A_i))$ (all $A_i$ have the same probability) because $\PR(O_{-1}(1)=O_{-1}'(1))=1$.        
        \item[(D)] Define 
        \[
            \Delta_i(t,s)=H_2(i,t,s)-H'_2(i,t,s)=:O_i(t)-O_i(s)-(O_i'(t)-O'_i(s))~.
        \]
        Note that $\Delta_i$ can only ever be non-zero when 
        \begin{itemize}
            \item[(i)]  one of the increments is non-zero 
            \item[(ii)] $O_{i-1}(1)\neq O_{i-1}'(1)$
        \end{itemize}
        Thus 
        \[
            \{\Delta_i(t,s)\neq 0\}\subset \{O_{i-1}(1)\neq O_{i-1}'(1)\}\cap \Bigg(\bigcup_{k=1}^\infty\{S_{i,k}\in[s,t] \}\cup \bigcup_{k=1}^\infty\{T_{i,k}\in[s,t] \}\Bigg)
        \]
        Using that
        \[
            \|H_2(t,s,\mathcal{F}_i)-H_2(t,s,\mathcal{F}_i') \|_2^2\leq 4\PR(\Delta_i(t,s)\neq 0)
        \]
        and independence of the $\eta_i$ we hence only need to show that
        \[
            \PR(O_{i-1}(1)\neq O_{i-1}'(1))\leq \chi^i, \quad \PR\Big(\bigcup_{k=1}^\infty\{S_{i,k}\in[s,t]\} \Big)+\PR\Big(\bigcup_{k=1}^\infty\{T_{i,k}\in[s,t]\} \Big)\lesssim |t-s|~.
        \]
        We have already established the first condition during the proof of (C). The second conditions follows by similar but easier arguments than those for (B).
    \end{itemize}
\end{proof}

\subsection{Proofs: Multiscale Statistics}
We will only provide proofs for the statistics \eqref{e:hatM}, the results for \eqref{e:hatR} follow by using the bound
\[
    \widehat R_N\leq \max_{1 \leq j \leq p}\max_{1 \leq l<k\leq N}\frac{\sum_{i=l+1}^k\Big((X_iO_i)(t_j)-\E[(X_iO_i)(t_j)]\Big)}{\sqrt{k-l}\log^2(eN/(k-l))\hat \pi_N(t_j)}
\]
which is valid whenever $H_{0}^\Delta$ holds. From thereon one can proceed exactly as for the test based on \eqref{e:hatM}. Before we continue we define the auxiliary quantity
\[
    M_N(t)=\max_{1 \leq l< k  \leq N-l}\frac{U_{k-l}^k(t)-U_k^{k+l}(t)}{\pi(t)\sqrt{l}\log^2(eN/l)}~.
\]
and note that Theorem \ref{theo:local:inference} is an immediate corollary of Theorem \ref{theo:gauss:approx:multiscale} and Lemma \ref{lem:empirical:count:multiscale} which are proven below. Note that under $H_0^{stat}$ we may assume WLOG that $\mu_i(t)=0$ for all $t$ and $i$.

\begin{theo}
\label{theo:gauss:approx:multiscale}
    Let $p=p(N)$ be a sequence satisfying $p \sim N^{C^*}$ and let $t_i=i/p\in [0,1], 1 \leq i \leq p$ be an array of points. Then there exists a sequence of mean-zero Gaussian vectors $W_N \in \R^p$ and a $C>0$ such that 
    \begin{align*}
        \sup_{x\geq C}&|\PR(\max_{(j,l,k)\in A}|W_{N,j,k,l}|>x)-\PR(\max_{1 \leq j \leq p}|M_N(t_j)|>x)|=o(1)\\
        &\C[W_{N,j_1,k_1,l_1},W_{N,j_2,k_2,l_2}]=\frac{\sigma_{\textbf{Y}}(t_{j_1},t_{j_2})\mathcal{L}(k_1,k_2,l_1,l_2)}{\pi(t_{j_1})\pi(t_{j_2})\rho_{2}(l_2/N)\rho_{2}(l_1/N)}~.
    \end{align*}
    where 
    \begin{align*}
         A&=\{(j,l,k)| \ 1 \leq j \leq p, l>N^{0.9}, l< k \leq N-l\}\\
        \rho_{2}(x)&=x^{1/2}\log^{2}(e/x)\\
        \mathcal L(k_1,k_2,l_1,l_2)
=&\ 
\big(\min\{k_1,k_2\}
     - \max\{k_1-l_1,\; k_2-l_2\}\big)_+ \\
&+
\big(\min\{k_1+l_1,\; k_2+l_2\}
     - \max\{k_1,\; k_2\}\big)_+ \\
&-
\big(\min\{k_1,\; k_2+l_2\}
     - \max\{k_1-l_1,\; k_2\}\big)_+ \\
&-
\big(\min\{k_1+l_1,\; k_2\}
     - \max\{k_1,\; k_2-l_2\}\big)_+ .
    \end{align*}
    Further, letting $B_W$ be the Brownian motion associated to the Gaussian process $W_Y$ with covariance $\C[W_Y(t),W_Y(s)]=(\pi(t)\pi(s))^{-1}\sigma_{\textbf{Y}}(t,s)$, we have that
    \[
        \sup_{x\geq C}|\PR(\max_{(j,l,k)\in A}|W_{N,j,k,l}|>x)-\PR(\Phi(B_W)>x)|=o(1)~.
    \]   
\end{theo}
\begin{proof}
    We first define 
    \begin{align*}
        \mathcal{Y}_{i,l,k}(t)=\begin{cases}
            \frac{\sqrt{N}}{\sqrt{l}\log^2(eN/l)\pi(t)}Y_i(t), &\quad k-l+1 <i \leq k\\
            \frac{-\sqrt{N}}{\sqrt{l}\log^2(eN/l)\pi(t)}Y_i(t), &\quad k+1\leq i\leq k+l\\
            0, &\quad else
        \end{cases}
    \end{align*}
    and note whenever $\max_{1 \leq j \leq p}|M_N(t_j)|\geq C$ holds we have by Lemma \ref{lem:small:scale:negligible} that
    \begin{align*}
        M_N(t)= &\max_{1 \leq j \leq p}\max_{l> N^{0.9}}\max_{l<k\leq N-l}\Big|N^{-1/2}\sum_{i=1}^N\mathcal{Y}_{i,l,k}(t_j)\Big|+o_{\PR}(1)\\
        =: &\max_{1 \leq j \leq p}\max_{l> N^{0.9}}\max_{l < k \leq N-l}|M_N(l,k,t_j)|+o_{\PR}(1)~. 
    \end{align*}
    It therefore suffices to establish the desired conclusion for the non-vanishing part of the right hand side quantity. We now want to apply Theorem \ref{thm:zhang:cheng:extension} and thus need to check its assumptions, in particular we want to apply \eqref{eq:GAP:bound}. Assumption (2.1) follows from noting that for $l>N^{0.9}$ we have
    \begin{align}
    \label{eq:factor:bound}
        \frac{\sqrt{N}}{\sqrt{l}\log^2(eN/l)\pi(t)} \lesssim N^{1/20}~.
    \end{align}    
    The restriction on the dimension is immaterial as our index set $A=\{(j,l,k)| \ 1 \leq j \leq p,l>N^{0.9},l<k\leq N-l\}$ grows polynomially. The decay assumptions on $\Theta_{k,j,q}$ and $\theta_{k,j,q}$ are an immediate corollary of \eqref{eq:factor:bound}, the definition of $\mathcal{Y}_{i,l,k}$ and Assumption (C'). The variance bounds in Assumption (2.3) follow by the same arguments as those when checking \eqref{ass:2} in the proof of Theorem \ref{lem:gauss:approx:confidence}.\\
    
    Consequently we may apply the theorem to obtain that there exists a sequence of Gaussian vectors $Z\in \R^q, q=|A|$ that satisfies
    \begin{align*}
        \sup_{x\geq C}&|\PR(\max_{1 \leq j \leq p}|M_N(t_j)|>x)-\PR(\max_{(j,l,k)\in A}|Z_{j,l,k}|>x)|=o(1)\\
        &\C[Z_{j_1,k_1,l_1},Z_{j_2,k_2,l_2}]=\C[M_N(l_1,k_1,t_{j_1}),M_N(l_2,k_2,t_{j_2})]~.
    \end{align*}
    Using that 
    \[
        \C[Y_i(t_j),Y_k(t_m)]\lesssim \chi^{|i-k|}
    \]
    in combination with Lemmas \ref{lem:lrc:partial:sums} and \ref{lem:gauss:comparison} then yields that we may find Gaussian vectors $W \in \R^q$ such that (assuming $k_1\leq k_2$)
    \begin{align*}
        \sup_{x \geq C}&|\PR(\max_{(j,l,k)\in A}|W_{j,k,l}|>x)-\PR(\max_{(j,l,k)\in A}|Z_{j,l,k}|>x)|=o(1),\\
        &\C[W_{j_1,k_1,l_1},W_{j_2,k_2,l_2}]=\frac{\sigma_{\textbf{Y}}(t_{j_1},t_{j_2})\mathcal{L}(k_1,k_2,l_1,l_2)}{\pi(t_{j_1})\pi(t_{j_2})\rho_{2}(l_2/N)\rho_{2}(l_1/N)}~,
    \end{align*}    
    which establishes the first part. For the second part we note that Lemma \ref{lem:gauss:tight} implies that $W_Y$ is an element of $\mathcal{C}[0,1]$. We further observe that
    \[
            M_{j,k,l}=\frac{B_W(k/N)(t_j)-B_W((k-l)/N)(t_j)-\Big(B_W((k+l)/N)(t_j)-B_W(k/N)(t_j)\Big)}{\sqrt{l/N}\log^{2}(eN/l)}
    \]
    has the same covariance structure as $W_{j,k,l}$, i.e. we may realize $W$ as (point evaluations of) weighted differences of increments of $B_W$. Any Brownian motion on a separable Banach space $\mathbb{B}$ takes values in the space $\mathcal{H}^{1/2,\gamma}(\mathbb{B}),$ for any $\gamma>1/2$. Specialized to this situation we thus obtain that 
    \[
        (t,s,x)\to(B_W(t)(x)-B_W(s)(x))/\rho_{2}(|t-s|)
    \]
    is an element of $C([0,1]^3)$. As the array $(t_j,k/N,l/N)_{(j,k,l) \in A}$ is (asymptotically) dense in $\{(x,y)| 0\leq y \leq \min(x,1-x)\}$ we therefore obtain
    \[
        \sup_{x\geq C}|\PR(\max_{(j,l,k)\in A}|W_{j,k,l}|>x)-\PR(\Phi(B_W)>x)|=o(1)
    \]
    as desired.
\end{proof}

\begin{lem}
    \label{lem:empirical:count:multiscale}
    We have that
    \[
        \sup_{x \in \R}|\PR(\max_{1 \leq j \leq p}|M_N(t_j)|>x)-\PR(\max_{1 \leq j \leq p} |\widehat M_N(t_j)|>x)|=o(1)~.
    \]
    In particular, in view of Theorem \ref{theo:gauss:approx:multiscale}, it then also holds that
    \[
        \sup_{x \in \R}|\PR(\Phi(B_W)>x)-\PR(\max_{1 \leq j \leq p}|\widehat M_N(t_j)|>x)|=o(1)~.
    \]
\end{lem}
\begin{proof}
    Follows by the exactly the same arguments as those starting in equation \eqref{eq:empirical:count:confidence} where we use the bound
    \begin{align}
    \label{eq:count:to:prop:multiscale}
       \max_{1 \leq j \leq p} \sqrt{N}|M_N(t_j)-\widehat M_N(t_j)|=O_{\PR}(1)
    \end{align}   
    in place of Lemma \ref{lem:approx:Nx}. This bound can be obtained by exactly the same arguments as Lemma \ref{lem:approx:Nx}.
    
\end{proof}

\begin{lem}
\label{lem:small:scale:negligible}
    We have that for any $\kappa>0$ there exists a $C=C(\kappa)>0$ such that
    \[
        \PR\Big(\max_{l\leq N^{0.9}}\max_{l<k\leq N-l} \frac{U_{k-l}^k(t)-U_k^{k+l}(t)}{\pi(t)\sqrt{l}\log(eN/l)^{2}}>C\Big)\leq N^{-\kappa}~.
    \]
    In particular, by the union bound we may always find a $C>0$ such that
    \[
        \PR\Big(\max_{l\leq N^{0.9}}\max_{l<k\leq N-l} \frac{(|U_{k-l}^k(\cdot)-U_k^{k+l}(\cdot))/\pi(\cdot)|_{\infty,\mathcal{T}}}{\sqrt{l}\log(eN/l)^{2}}>C\Big)=o(1)
    \]
    for any $\mathcal{T}$ that satisfies $|\mathcal{T}|\leq N^{C^*}$.
\end{lem}
\begin{proof}
Recall the filter formalism in and after eq. \eqref{eq:filter:formalism}. For this proof, when considering the dependence measures $\delta_q$ and $\Theta_q$ of a filter $H$, we will make the associated filter explicit by writing $\delta_q(H)$ or $\Theta_q(H)$. First note that we may write $Y_i(t)=G_1(t,\mathcal{F}_i)G_2(t,\mathcal{F}_i)=:G_3(t,\mathcal{F}_i)$, defining $Y_i'$ in the obvious way we hence obtain
    \[
        \|Y_i(t)-Y_i'(t)\|_q\leq \sqrt{q}\|X_i(t)\|_{\Psi_2}\delta_q(G_2,i)+\delta_q(G_1,i)
    \]
    so that
    \[
        \Theta_q(G_3)\lesssim\sum_{i=1}^\infty(\sqrt{q}\delta_q(G_2,i)+\delta_q(G_1,i)) \lesssim \sqrt{q}\Theta_q(G).
    \]
    By Assumption (C') we therefore obtain that
    \[
        \sup_{q \geq 2}q^{-1}\Theta_q(G_3)<\infty
    \]
    so that by Theorem 3 from \cite{Wu:Wu:2016} it holds for some $C_1,C_2$ not depending on $l,k$ or $t$ that
    \[
        \PR( (k-l)^{-1/2}|U_l^k(t)|>x)\leq C_1\exp(-C_2x^{1/2})~.
    \]
    In particular, choosing  $x=y(KC_2^{-1})^{2}\log^2(eN/l)$ we have that
    \[
    \PR\Big( \frac{|U_l^k(t)|}{\sqrt{l}\log^2(eN/l)}>y(KC_2^{-1})^{2}\Big)\leq C_1\Big(\frac{l}{Ne}\Big)^{Ky^{1/2}}~.
    \]
    Letting, for fixed $t$, $W_N(k/N)=N^{-1/2}U_0^k(t)$ as well as $\rho_2(x)=x^{1/2}\log^{2}(e/x)$ we are hence in the setting of Theorem A.3 from \cite{koehne:mies:2025} where the constants $C, \kappa(y)$ in its statement satisfy $C=(KC_2^{-1})^{2}$ and $\kappa(y)\geq K$. In its proof the authors derive, for sufficiently large $m$ and $t>C$, the inequality
    \begin{align*}
        \Pr\Big(\sup_{|k/N-l/N|\leq 2^{-m(1-\delta)}}\frac{|W_N(k/N)-W_N(l/N)|}{\rho_2((k-l)/N)}>2y\Big) \leq 2^{m(1+\delta)-m(1-\delta)\kappa(y)}\leq 2^{m(3\delta+1-K)}~.
    \end{align*}
    for any $\delta<(K-1)/(K+1)$. Setting $m=\log_2(N)$ we hence obtain
    \begin{align*}
       \Pr\Big(\sup_{|k-l|\leq N^{\delta}}\frac{|U_l^k(t)|}{\rho_2((k-l)/N)}>y\Big)\leq N^{3\delta+1-K}~.
    \end{align*}
    for any $y>2C$. Finally, letting $\delta=0.9$ and setting $K$ sufficiently high we obtain the desired result by an application of the triangle inequality.
\end{proof}

\subsection{Proofs: Local Alternatives}

In this Section, we give the proof of Theorem \ref{thm:loc:cons}. Before that, we provide some additional comments on the theorem.

\begin{rem}(Details on Theorem \ref{thm:loc:cons}) \label{rem:loc:cons}\\ 
 $ {}\quad {}$ \indent \textit{I) Multiscale detection:} The power of the multiscale statistics is most easily visible in the case of the stationarity test. If $b_N =b>0$ (fixed signal size) we have consistency as long  as $a_N \gg \log^2(N)/N$. Since consistency is only possible as long as $a_N\gg 1/N$ we see how close the multiscale result is to optimality on small scales. Conversely, on large scales where $a_N=a>0$, we have a sample size for the signal of $\mathcal{O}(N)$ and detection is possible as long as $b_N\gg N^{-1/2}$ which is of course optimal.\\
 $ {}\quad {}$ \indent \textit{II) Uniformity:} The aim of our investigation under the alternatives is to give an insight into the power of multiscale statistics, not to provide an investigation in full generality. For the interested reader, we refer to our below proofs, which clearly show that uniform consistency of our tests is possible over large classes of local alternatives that satisfy roughly our $a_N,b_N$ conditions. In particular, there exists uniformity regarding the direction $\mu^*(\cdot)$ of the alternative, as long as $\mu^*(\cdot)$ satisfies some  smoothness property. We have avoided this more general presentation in the main body of our paper, as it is harder to read and adds little insight.
 
\end{rem}
\begin{proof}
    We first consider the alternative \eqref{e:alt:1}. It suffices to show that $\widehat M_N(t_{max})$ with $t_{max} \in \arg \max_{1 \leq i \leq p} |\mu^*(t_i)|$ diverges in probability. To that end we note that by  \eqref{eq:count:to:prop:multiscale} and Lemma \ref{lem:empirical:count:multiscale}
    \begin{align*}
         &\max_{1 \leq l < k \leq N-l}\frac{|\sum_{j=k-l+1}^k(X_jO_j)(t_{max})-\sum_{j=k+1}^{k+l}(X_jO_j)(t_{max})|}{\hat \pi_N(t_{max})\sqrt{l}\log^{2}(eN/l)}\\
         \geq &\max_{1 \leq l < k \leq N-l}\frac{|\sum_{j=k-l+1}^k\mu_j(t_{max})-\sum_{j=k+1}^{k+l}\mu_j(t_{max})|}{\sqrt{l}\log^{2}(eN/l)}+O_{\PR}(1)
    \end{align*}
    As $C^*>1/(2\theta)$ we may replace $t_{\max}$ by the maximizer of $\mu^*(\cdot)$ incurring only a negligible error. Choosing $k=\lceil x_0N\rceil$ and $l=\lfloor a_NN\rfloor$ then yields that the deterministic part on the right hand side is at least of order
    \[
        b_N\frac{a_NN}{\sqrt{a_NN}\log^{2}(1/a_N)},
    \]
    whenever $a_n \to 0$, which yields the desired result.\\

    For the class \eqref{e:alt:2} let $t^\star$ satisfy $\mu(t^\star)=\Delta$ (or choose a grid point with
    $\mu(t_j)\ge \Delta-o(1)$ by Hölder continuity and the condition on $C^*$). Under \eqref{e:alt:2},
    \[
    \mu_i(t^\star)=\Delta+b_N\,d\!\left(\frac{|i/N-x_0|}{a_N}\right).
    \]
    By $d(0)-d(x)\sim x^\kappa$ there exists $a>0$ such that for all $0\le x\le a$,
    $d(x)\ge 1-2x^\kappa$.
    Choose $c>0$ small enough and set
    \[
    m:=\left\lfloor c\,N a_N b_N^{1/\kappa}\right\rfloor,
    \qquad
    l_0:=\left\lfloor Nx_0-\frac{m}{2}\right\rfloor,\quad k_0:=l_0+m.
    \]
    Then for all $j\in\{l_0+1,\dots,k_0\}$ we have
    $\frac{|j/N-x_0|}{a_N}\le c\, b_N^{1/\kappa}\le a$, hence
    $d\big(\frac{|j/N-x_0|}{a_N}\big)\ge 1-2c^\kappa b_N$ and therefore,
    for $N$ large,
    \[
    \mu_j(t^\star)-\Delta
    =(\Delta+b_N)d(\cdot)-\Delta
    \ge b_N-2(\Delta+b_N)c^\kappa b_N
    \ge \frac{b_N}{2}.
    \]
    Consequently,
    \[
    \max_{1\le l<k\le N}
    \frac{\sum_{j=l+1}^k \mu_j(t^\star)-(k-l)\Delta}{\sqrt{k-l}\log^2(eN/(k-l))}
    \ \ge\
    \frac{(b_N/2)\,m}{\sqrt{m}\log^2(eN/m)}
    \ \gtrsim \frac{b_N\sqrt{m}}{\log^2(eN/m)},
    \]
    which yields the desired conclusion.
    
\end{proof}

\subsection{Extensions of Theorem 2.1 from \cite{zhang:cheng:2018}}

We state the extension with the terminology given in \cite{zhang:cheng:2018} and refer the interested reader to that paper for all the required definitions. We also note that we make no effort to optimize the dependence of the rates on any of the involved parameters as this is not needed for our purposes. \\

We now state the adjusted assumptions, using the same numbering as in \cite{zhang:cheng:2018} for comparability. We also adopt all notations from that paper.

\begin{enumerate}
    \item[(2.1)] Assume that there exists $\mathfrak{D}_N,K$ such that 
    \[
        \max_{1 \leq i \leq N}\max_{1 \leq j \leq p}\E[\exp(|x_{i,j}|/\mathfrak{D}_N)]\leq K
    \]
    \item[(2.2)] We have that there exists $M=M(N), \gamma=\gamma(N), C_4$ such that
    \[
        N^{3/8}M^{-1/2}\ell_n^{-5/8}\geq C_4\max\{\mathfrak{D}_N\ell_N,\ell_N^{1/2}\}
    \]
    where $\mathfrak{D}_N$ is given in (2.1) and $\ell_n:=\log(Np/\gamma)$.
    \item[(2.3)] For some constants $c_2,c_3$  and sequences $\tilde c_N,d_N>0$ we have
    \begin{align*}
        \tilde c_N\lesssim\min_{1 \leq j \leq p}\sigma_{j,j}&\leq \max_{1 \leq j \leq p}\sigma_{j,j}<c_2\\        
       \max_{1\leq j \leq p}\sum_{k=1}^\infty  k\theta_{k,j,3}&<d_N
    \end{align*}
\end{enumerate}
We also remark that the filters $\mathcal{G}_i$ in \cite{zhang:cheng:2018} are a triangular array and that the quantities $\Theta_{j,k,l}$ that will also appear in the following theorem are thus also implicitly dependent on $N$.

\begin{theo}
\label{thm:zhang:cheng:extension}
    Suppose that Assumptions (2.1) to (2.3) above hold. Then, for any $q\geq 2$,  we have
    \begin{align}
    \label{eq:gauss:approx:errors}
        \rho_N\lesssim &N^{-1/8}M^{1/2}\ell_n^{11/8} (\tilde c_N^{-1/2}\lor d_N)\\
         &+\gamma+\Big(N^{1/8}M^{-1/2}\ell_N^{-3/8}\Big)^{q/(1+q)}\Big(\sum_{j=1}^p\Theta^q_{M,j,q}\Big)^{1/(1+q)}\\
         + & \Xi_M^{1/3}(1\lor\log(p/\Xi_M))^{2/3}\tilde c_N^{-1/2}
    \end{align}
    where $\Xi_M=\max_{1 \leq j \leq p}\sum_{k=M}^\infty k\theta_{k,j,2}$. In particular, assuming that
    \begin{itemize}
        \item[(i)] Assumption 2.1 in \cite{zhang:cheng:2018}  holds with $\mathfrak{D}_N\lesssim N^{(3-17b)/8}$
        \item[(ii)] $p \lesssim \exp(N^b)$ for $0<b<1/11$
        \item[(iii)] $\max_{1 \leq j \leq p} \Theta_{u,j,q}\leq \mathfrak{D}_N\chi^u$ for some $\chi<1$ and $q\geq 2$
        \item[(iv)] Assumption 2.3  in \cite{zhang:cheng:2018}  holds and $\tilde c_N\gtrsim \log(N)^{-2k}$ for some $k>0$
    \end{itemize}
    we have that
    \begin{align}
    \label{eq:GAP:bound}
        \rho_N\lesssim N^{-(1-11b)/8}(\log(N)^k\lor d_N)
    \end{align}   
    by choosing $M=CN^b$ for sufficiently large $C$ and $\gamma=N^{-(1-11b)/8}$.
\end{theo}
\begin{proof}
    The proof is divided into 5 steps. Due to its length and because the changes require only minor adaptations we only indicate the required changes for each step.

    \begin{enumerate}
        \item[Step 1:] When following the arguments in Proposition A.1 that lead to equation (39) the term $(\beta^{-1}\log(p)+\psi^{-1})\sqrt{1 \lor \log(p\psi)}$ needs to be replaced by 
        \[
            (\beta^{-1}\log(p)+\psi^{-1})\sqrt{1 \lor \log(p\psi)}\tilde c_N^{-1/2}~.
        \]
        This follows because instead of using Lemma 2.1 in \cite{chernozhukov:chetverikov:kato:2013} we use Theorem 1 in \cite{chernozhukov:chetverikov:kato:2017} which gives an explicit dependence of the bound on $\tilde c_N$.
        \item[Step 2:] The bounds on  $\sum_{k=1}^\infty\theta_{k,j,q}$ get worse by a factor of $d_N$, yielding an extra factor of $d_N$ in the final bound on $\varphi^{(M)}(M_x)$
        \item[Step 3:] Same as in Step 2, i.e. $\varphi^{(M)}(M_y)=d_NC'/M_y^2$ for some $C'>0$.
        \item[Step 4:] Due to the worse bound in Step 1 the upper bound in equation (49) is now given by $n^{-1/8}M^{1/2}\ell_n^{11/8}$. Equations (46) and (48) pick up an extra $d_N$ factor due to the worse bounds in Step 2 and 3.
        \item[Step 5:] The bound in equation (52) incurs an extra $\tilde c_N^{-1/2}$ factor as we apply Lemma \ref{lem:gauss:comparison} instead of Theorem 2 from \cite{chernozhukov:chetverikov:kato:2015}. Consequently the bound in equation (52) becomes 
        \[
            \Xi_M^{1/3}(1\lor\log(p/\Xi_M))^{2/3}\tilde c_N^{-1/2}~.
        \]
        Also, when using arguments similar to those in Lemma A.2 one uses the same substitutions of the bounds involving $\sum_{k=1}^\infty j\theta_{k,j,q}$ as in Step 2 and 3.
    \end{enumerate}
\end{proof}

\begin{lem}
\label{lem:gauss:comparison}
    Let $Z_1,Z_2 \in \R^p$ be two centered Gaussian vectors satisfying $\min_{1 \leq i \leq p}\min\{\C[Z_{1i},Z_{1i}],\C[Z_{2i},Z_{2i}]\}>c_N$. Then
    \[
        \sup_{x \in \R}|\PR(\max_{1 \leq i \leq p}Z_{1i}>x)-\PR(\max_{1 \leq i \leq p}Z_{2i}>x)|\lesssim \Delta^{1/3}(1 \lor \log(p/\Delta))^{2/3}c_N^{-1/2}
    \]
    where 
    \[
        \Delta=\max_{1 \leq i,j\leq p}|\C[Z_{1i},Z_{1j}]-\C[Z_{2i},Z_{2j}]|~.
    \]
\end{lem}
\begin{proof}
    This is a slightly more quantitative version of Theorem 2 from \cite{chernozhukov:chetverikov:kato:2015}. The proof follows by exactly the same arguments, the only difference being the application of Theorem 1 from \cite{chernozhukov:chetverikov:kato:2017} instead of Theorem 3 from \cite{chernozhukov:chetverikov:kato:2015}.
\end{proof}

\subsection{Neccessary and sufficient conditions for SCBs}
In this section we consider the data model 
\[
    X_i=\mu+\varepsilon_i, \quad i=1,...,N
\]
where as in \eqref{eq:filter:formalism} we assume $\epsilon_i$ to be represented by a filter $G_1(t,\mathcal{F}_i)$.  We assume here that $X_i$ is fully observed and allow $\epsilon_i\in L^\infty[0,1]$, i.e. bounded but discontinuous errors are included. The proofs can easily be extended to the partially observed setting - we restrict ourselves in this way to make clear that our results are not an artifact of the additional difficulties imposed by partially observed data. We make the following assumptions:
 \begin{enumerate}
        \item[(S1)] (Moments) It holds that $\E[|\varepsilon_1(\cdot)|_\infty^4]< \infty$.
       \item[(S2)] (Continuity) For some constant $\vartheta \in (0,1]$ the function $\mu(\cdot)$ is $\vartheta $-Hölder continuous.
        \item[(S3)] (Weak dependence of functions) Define the dependence measure
		\begin{align} \label{2.7}
			\delta_\varepsilon(i):=\sup_{t\in [0,1]}\|\varepsilon_i(t)-\varepsilon_i'(t)\|_{3}, 
		\end{align}
		There exists a constant  $\chi\in(0,1)$ such that 
			$\delta_{\varepsilon}(i)=\mathcal{O}(\chi^i)~$ for $i\geq 0$ .
 		\item[(S4)]  (Long-run variances) We define the following long-run variance 
		\begin{align} \label{lrv}
			\sigma_\varepsilon(t,s):=\sum_{k=-\infty}^\infty \C[\varepsilon_0(t),\varepsilon_k(s)]
		\end{align}
     and assume that $\sigma_\varepsilon^2(t):=\sigma_\varepsilon(t,t)$ is bounded away from $0$, uniformly over $t$. 
     \item[(S5)] Let $W=\{W(t)\}_{t \in [0,1]}$ be the centered Gaussian process with covariance $\sigma_\varepsilon(\cdot,\cdot)$. Define the canonical metric $d(s,t)\;:=\;\bigl(\mathbb{E}(W_s-W_t)^2\bigr)^{1/2}$ and assume that 
\[
\gamma_2(T,d)
\;:=\;
\inf_{\{\mathcal{A}_n\}}
\ \sup_{t\in T}\ \sum_{n\ge 0} 2^{n/2}\,\operatorname{diam}\!\bigl(A_n(t)\bigr)<\infty.
\]
Here, the infimum is taken over all admissible sequences of partitions
$\{\mathcal{A}_n\}_{n\ge 0}$ of a set $T$, meaning that $|\mathcal{A}_0|=1$ and
\[
|\mathcal{A}_n|\ \le\ 2^{2^n}\qquad\text{for all }n\ge 1.
\]
For each $t\in T$, $A_n(t)$ denotes the unique element of $\mathcal{A}_n$
containing $t$, and
\[
\operatorname{diam}(A)\;:=\;\sup\{d(s,u): s,u\in A\}.
\]
	\end{enumerate}
These are essentially assumptions (A) to (E), where (D) has been replaced by a condition characterizing when the Gaussian process $W$ used to obtain the asymptotic quantiles (of $|W|_\infty$) is finite almost surely.  This condition is obtained from Talagrand's majorizing measure theorem (Theorem 2.1.1 in \cite{Talagrand:2005}). Assumption (S5) is thus our way of imposing regularity on the error process; a substantially weaker version than the the smoothness assumptions on the errors typically imposed in other works. Indeed, it is substantially weaker and only concerns boundedness of the Gaussian process associated to the covariance of the error process. The following theorem shows that, assuming (S1) to (S4), (S5) is sufficient to guarantee existence of SCBs.
  \begin{theo}
    \label{theo:confidence:bands:fully}
       Assume that (S1)-(S5) hold and that $p \sim  N^{C^*} $ for some $C^*>\frac{1}{2\vartheta}$. Let $W=\{W(t)\}_{t \in [0,1]}$ be the centered Gaussian process with covariance $\sigma_\varepsilon(\cdot,\cdot)$.
       Then $W$ is almost surely bounded. Moreover, denoting as before the $(1-\alpha)$-quantile of $|W(\cdot)|_\infty$ by $q_{1-\alpha}$, we have that 
        $\{\hat C^\pm(t)\}_{t \in [0,1]}$ (defined in \eqref{e:int:sing} and \eqref{e:int:SCB}) is an asymptotic $(1-\alpha)$ confidence band for $\mu(\cdot)$.        
    \end{theo}
\begin{proof}
    Follows by exactly the same arguments as Theorem \ref{theo:confidence:bands}, the only difference being that instead of Lemma \ref{lem:gauss:tight} we now need to show that $W$ is bounded almost surely. This follows immediately from Theorem 2.1.1 in \cite{Talagrand:2005}.
\end{proof}
Conversely, let us assume that $(\varepsilon_i)$ satisfies (S1) to (S4) but not (S5). Further assume that the associated Gaussian process $W$ attains its supremum over a countable set $(t_i)_{i \in \N}$.  Note that existence of this countable set is weaker than boundedness, i.e. it does not imply (S5).   We will now show that any polynomially growing enumeration of that grid does not allow for SCBs, i.e. it holds that for any supremum-attaining sequence $(t_i)_{i\in \N}$ and any $p=p(N)\to \infty, p(N)\lesssim N^c$ for some $c>0$, we have for any bounded function $q(\cdot)$ that
\[
    \PR( \mu(t_i) \leq \hat \mu(t_i)+q(t_i)N^{-1/2}, i=1,...,p)=o(1).
\]
I.e. there is no way to pick a boundary of width $N^{-1/2}$ based on $\hat \mu(\cdot)$ such that it covers $\mu$ with non-negligible probability.

To see this note that under (S1) to (S4) the Gaussian approximation 
\[
    \PR( \sqrt{N}(\hat \mu(t_i)-\mu(t_i))\leq q(t_i), i=1,...,p)=\PR( W_i\leq q(t_i), i=1,...,p)+o(1),
\]
is valid, where $\C(W_i,W_j)=\sigma_\varepsilon(t_i,t_j)$ (see the proof of Theorem \ref{theo:confidence:bands:fully}). By the choice of the sequence $t_1,...$ and Theorem 2.1.1 in \cite{Talagrand:2005} the right hand side probability is $o(1)$, yielding the desired result. \\

Finally we provide an example satisfying the assumptions of Theorem \ref{theo:confidence:bands} but for which 
\[
    \sqrt{N}|\hat \mu(\cdot)-\,\mu(\cdot)|_\infty
\]
is convergent to $\infty$. We will give an example in the partially observed setting, but similar ideas can be used in the scenario of complete observations, too. For this purpose, we consider an extremely simple case:  We define $\varepsilon_i(\cdot)$ i.i.d. as constant zero functions, i.e. we observe the mean function $\mu(\cdot) \equiv 1$ without noise.  Let $(U_i)_{i\ge1}$ be i.i.d.\ $\mathrm{Unif}[0,1]$ random variables.
Let $(B_i)_{i\ge1}$ be i.i.d.\ $\mathrm{Bernoulli}(1/2)$, independent of $(U_i)$. Let $\mathcal{K}$ denote the collection of all finite subsets of $[0,1]$.
Since $|\mathcal{K}|=|[0,1]|$, fix an indexing
\[
[0,1] \ni t \mapsto A_t \in \mathcal{K}.
\]
Define
\[
O_i(t) := B_i + (1-B_i)\mathbf 1\{U_i \in A_t\},
\qquad t\in[0,1],
\]
which are independent as well as bounded and thus satisfy (A), (C) and (E). For any fixed $t$,
\[
\mathbb{E}[O_i(t)]
= \mathbb{E}[B_i]
  + \mathbb{E}[1-B_i]\mathbb{P}(U_i \in A_t).
\]
Since $A_t$ is finite, $\mathbb{P}(U_i \in A_t)=0$ and thus $\mathbb{E}[O_i(t)] = \frac{1}{2}$. Therefore $t \mapsto \mathbb{E}[O_i(t)]$ is constant, Lipschitz continuous, and strictly positive, i.e. it satisfies (B). It is straightforward to check that, for any pair $s,t \in [0,1]$
\[
\Pr(O_i(t)\neq O_i(s))=0
\]
which immediately implies (D). Now define the random finite set
\[
F_n := \{X_1,\dots,X_n\}.
\]
By construction of the class $\{A_t\}$, there exists $t_n \in [0,1]$ 
such that $A_{t_n} = F_n$. For this choice of $t_n$ and for every $i\le n$, $\mathbf 1\{X_i \in A_{t_n}\} = 1$, so $O_i(t_n)= B_i + (1-B_i)\cdot 1= 1$. Hence
\[
\frac{1}{n}\sum_{i=1}^n O_i(t_n) = 1.
\]
But $\mathbb{E}[O_i(t_n)] = 1/2$, therefore
\[
\sup_{t\in[0,1]}
\left|
\frac{1}{n}\sum_{i=1}^n \bigl(O_i(t) - \mathbb{E}[O_i(t)]\bigr)
\right|
\;\ge\;
\left|1 - \frac{1}{2}\right|
= \frac{1}{2},
\]
i.e. uniform convergence always fails.

\subsection{Adapting ~(\ref{e:hatR})}
\label{sec:app:adaption}
As indicated in Remark \ref{rem:details} II), we now construct an adapted version of the test \eqref{e:hatR}. More precisely, we suggest a way to obtain tighter quantiles which yield increased power. The approach uses, roughly speaking, that distributions of supremum statistics for functional data only depend on a subset of indices. So the maximum of the limiting process is only attained over a subset of indices and this yields smaller values for the limit than a maximum over all possible limits. Systematic discussions of this approach can be found in \cite{Buecher:Dette:Heinrichs:2021}, \cite{dette:kokot:2022}, and here we merely give an adaptation to our current multiscale statistic.
To that end, we define the following set estimators
\[
    \hat{\mathcal{E}}:=\Bigg\{ 1 \leq l<k\leq N\Big| (k-l)^{-1}|U_l^k/\hat \pi_N|_{\infty,\mathcal{T}}>\Delta-a_1\bar \sigma_{\textbf{Y}}\log(N)/\sqrt{N}\Bigg\}
\]
and
\[
    \hat{\mathcal{E}}(l,k):=\Bigg\{ t \in \mathcal{T} \Big| \frac{U_l^k(t)}{\hat \pi_N(t)\sqrt{k-l}\log^{2}(eN/(k-l))}>\frac{|U_l^k/\hat \pi_N|_\infty-a_2\log(N)\bar \sigma N^{-1/2+a_2}}{\sqrt{k-l}\log^{2}(eN/(k-l))} \Bigg\}
\]
where $\bar \sigma_{\textbf{Y}}=1/p \sum_{j=1}^p\hat \sigma_{\textbf{Y}}(t_j,t_j)$ and $a_1$ and $a_2$ need to be chosen by the user. Following existing work (again most notably \cite{Buecher:Dette:Heinrichs:2021}, \cite{dette:kokot:2022}) we recommend using $a_1=0.2$ and $a_2=0.001$.\\
With these estimators for the extremal sets we now may construct the functional
\[
    \Psi_{\text{ext}}(H)=\max_{(l,k) \in \hat{\mathcal{E}}}\max_{t \in \hat{\mathcal{E}}(l,k)}\frac{H(k/N)(t)-H(l/N)(t)}{\sqrt{(k-l)/N}\log^2(eN/(k-l))}
\]
and denote by $\hat q_{1-\alpha}$ the $(1-\alpha)$ quantile of $\Psi_{\text{ext}}(B_{\widehat W})$ where $\widehat W$ is defined in \eqref{eq:defin:hatW}. We may then test the hypotheses \eqref{e:threshold}
via
\[
    \widehat R_{N}>\hat q_{1-\alpha}~.
\]

\subsection{Variance estimation for the multiscale statistic} \label{var:mult}
For the multiscale inference procedures we need a long-run covariance estimator that can handle smoothly and abruptly changing means in a high-dimensional setting. To that end we use the robust M estimator from Appendix A.1. in \cite{Li:Chen:Wang:Wu:2024} to estimate $\sigma_{\textbf{Y}}(s,t)$. We recall its precise definition for the convenience of the reader.

Group the $N$ observations into {$N_0=\lfloor N/m\rfloor-1$} blocks with each block having size $m\in\N$. We denote the index set of the observations in the $k$-th block by $M_k=\{km+1,...,(k+1)m\}$ ($k=0,...,N_0$) and define the average of the observations in the block $M_k$ as
$$\psi_k=\sum_{i\in M_k}Y_i/m.$$ Next, let $\psi_k=(\psi_{k,1},\psi_{k,2},\ldots,\psi_{k,p})^{\top}$, and for $k=1,\ldots,N_0$, we define 
\begin{equation}
    \label{eq:longrun:element}
    \hat\sigma_{i,j,k}=(m/2)(\psi_{k,i}-\psi_{k-1,i})(\psi_{k,j}-\psi_{k-1,j}).
\end{equation}
Now,  for some $\alpha_{i,j}>0$, we define the zero-objective function
\begin{equation}
    \label{eq_longrun_zero}
    h_{i,j}(u)=\sum_{k=1}^{N_0}\frac{1}{N_0}\phi_{\alpha_{i,j}}(\hat\sigma_{i,j,k}-u),
\end{equation}
where $\phi_{\alpha}(x)=\alpha^{-1}\phi(\alpha x)$ and 
\begin{equation}
    \label{eq_longrun_phi}
    \phi(x)= \begin{cases}
        \log(2), \qquad & x\ge 1,\\
        -\log(1-x+x^2/2), \qquad & 0\le x< 1,\\
    {\log(1+x+x^2/2)}, \qquad & -1\le x< 0,\\
        -\log(2), \qquad & x< -1.\\
    \end{cases}
\end{equation}
The function $|\phi(x)|$ is bounded by $\log(2)$, and it is also Lipschitz continuous with the Lipschitz constant bounded by 1. Additionally the influence function $\phi(x)$ has neat envelopes given by
\begin{equation}
    -\log(1-x+x^2/2) \le \phi(x) \le \log(1+x+x^2/2).
\end{equation}

To obtain a consistent estimator of the long-run covariance matrix $(\sigma_{\textbf{Y}}(t_i,t_j))_{1 \leq i,j \leq p}$, we solve the equation $h_{i,j}(u) = 0$ for each $1\le i,j\le p$, and use the root to be the estimator of the element with indices $(i,j)$ in the long-run covariance matrix. We denote this estimator by $\hat\sigma_{\textbf{Y}}(t_i,t_j)$. In the case t hat the equation $h_{i,j}(u) = 0$ has more than one solution, we may pick from them at random to define $\hat\sigma_{\textbf{Y}}(t_i,t_j)$. Finally, we collect the estimates of all the elements and combine them into the long-run covariance matrix,
\begin{equation}
    \label{eq_longrun_est}
    \hat \sigma_{\textbf{Y}}=(\hat\sigma_{\textbf{Y}}(t_i,t_j))_{1\le i,j\le p}
\end{equation}
In particular, we define $\bar\sigma_i=2\sum_{\lfloor N_0/4\rfloor \le k\le \lfloor 3N_0/4\rfloor}\sqrt{\hat\sigma_{i,i,k}}/N_0$ and set $\alpha_{i,j}$ in expression (\ref{eq_longrun_zero}) to be $\bar\sigma_i\bar\sigma_j\sqrt{m/N}$. Following \cite{Li:Chen:Wang:Wu:2024} we set $m=\sqrt{N/\log(Np)}$. The estimate for $\widehat \C$ is then again obtained by setting
\[
    \widehat \C(t_i,t_j)=\frac{N^2}{\hat N(t_i)\hat N(t_j)}\hat \sigma_{\textbf{Y}}(t_i,t_j)~.
\]

\subsection{Graphical details for the data analysis in Section \ref{sec:polution}} 
We provide some visual illustrations pertaining to our pollution data analysis in Section \ref{sec:polution}. 
  At each time of the day $t$ and day of the year $d$, we assign color if $r(t,k,l)>q_{1-\alpha}$, i.e. if the tuple $(t,d)$ is part of a significant scale. We put darker red colors for shorter scales $(k-l)$ and lighter yellow colors for larger $(k-l)$. Darker red colors indicate that we have high certainty that a threshold exceedance happened close to $(t,d)$ and yellow colors less certainty about the precise location. The heatmaps for $\Delta=50, 90, 180$ are gathered in Figure \ref{fig:exceedances}. 

\begin{figure}[H] 
  \centering
  \includegraphics[width=0.4\textwidth]{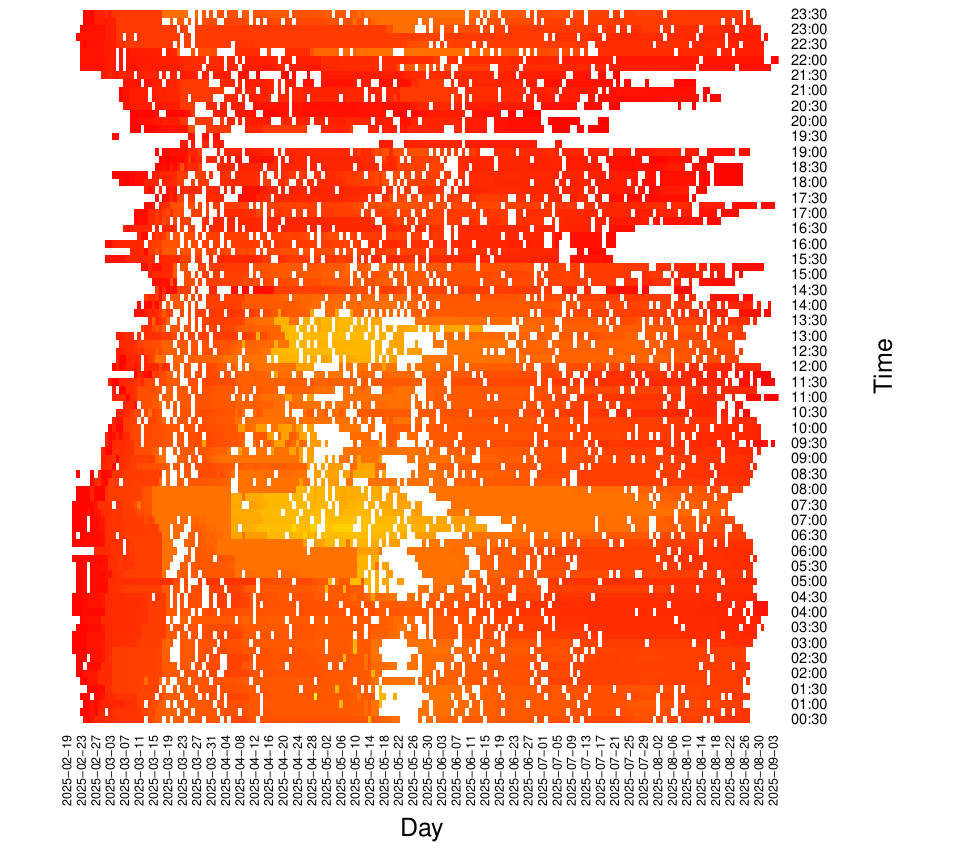}
  \includegraphics[width=0.4\textwidth]{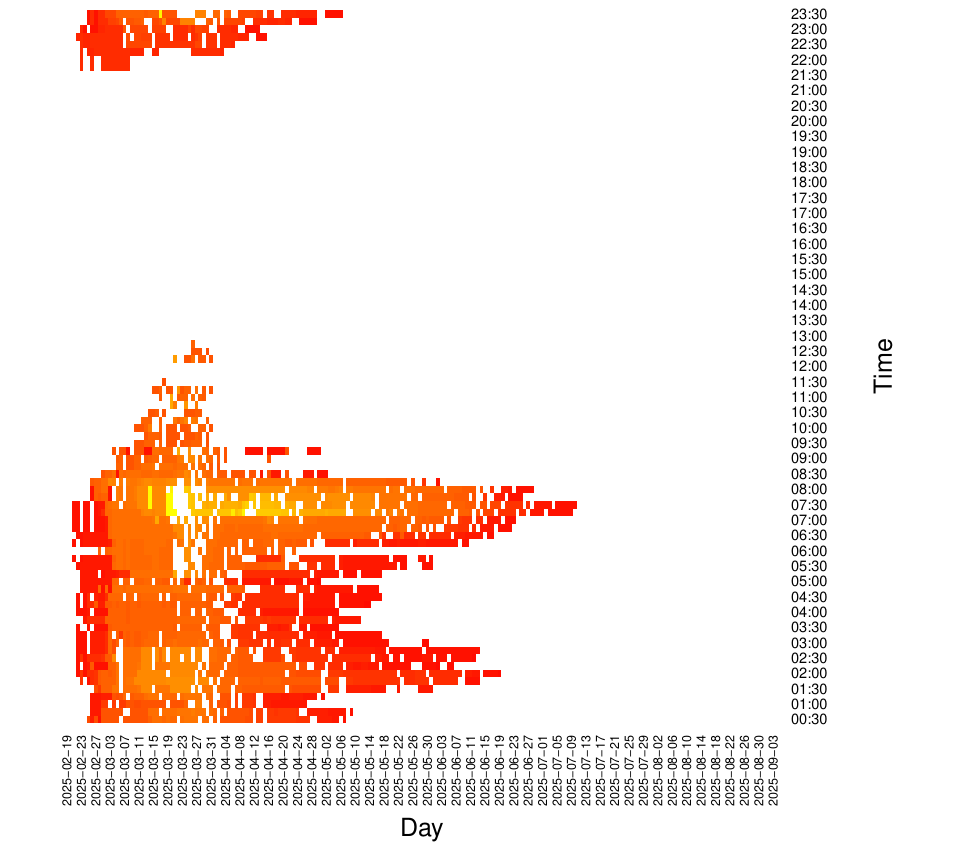}
  \includegraphics[width=0.4\textwidth]{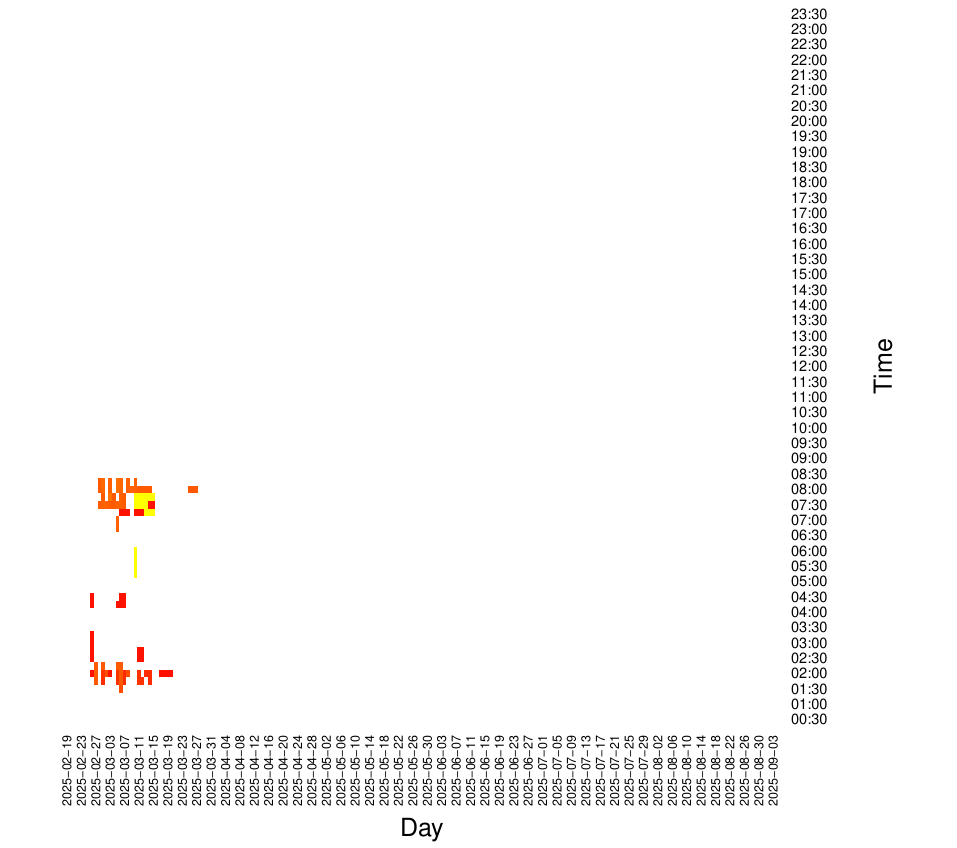}
  \includegraphics[width=0.4\textwidth]{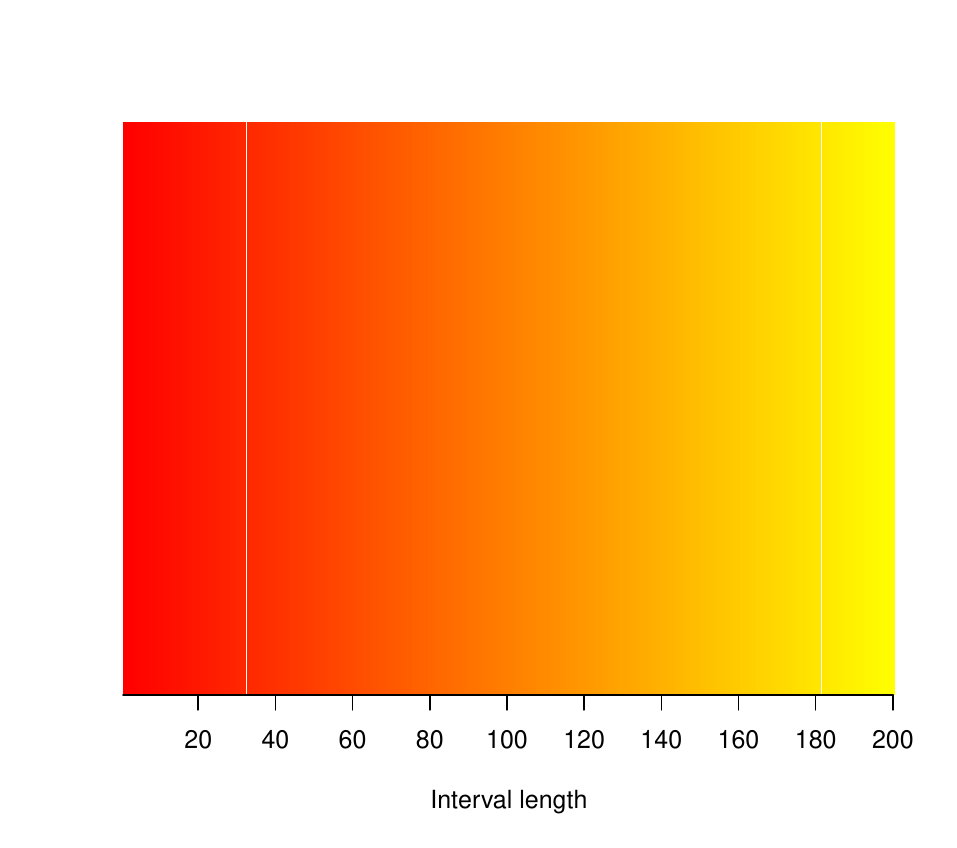}
  \caption{Heatmaps of exceedances for different thresholds $\Delta$ (upper left: 50, upper right: 90, lower left: 180). The color at a point encodes the total length of the smallest scale centered at that point for which the test statistic was significant.   \label{fig:exceedances}
  }

\end{figure}

\putbib
\end{bibunit}

\end{document}